\documentclass[11pt]{article}
\usepackage{tikz}
\usetikzlibrary{positioning} % 用于节点相对定位
\usepackage{graphicx}%插图
\usepackage{float}%浮动
\usepackage{subfigure}%插入多图
\usepackage{caption}
\usepackage{multirow}
\usepackage{makecell}
\usepackage{appendix}
\usepackage[figuresright]{rotating}
\usepackage{booktabs}
\usepackage[linesnumbered,ruled]{algorithm2e}
\usepackage[margin=1in]{geometry}
\usepackage{hyperref}
\usepackage{amsfonts}
\usepackage{mathrsfs}
\usepackage{comment}
\usepackage{amsmath}
\usepackage{amssymb}
\usepackage{amsthm}
\usepackage{amscd}
\usepackage{indentfirst}
\usepackage[all]{xy}
\usepackage{titlesec}
\usepackage{enumerate}
\usepackage{bm}
\usepackage{enumitem}
\usepackage{color}
\usepackage{threeparttable}
\usepackage{verbatim}
\usepackage{dsfont}
\usepackage{arydshln}
\usepackage{makecell}
\usepackage{fancyhdr}
\newtheorem{theorem}{Theorem}[section]
\newtheorem{lemma}[theorem]{Lemma}
\newtheorem{proposition}[theorem]{Proposition}

\newtheorem{conjecture}[theorem]{Conjecture}
\newtheorem{definition}[theorem]{Definition}

\newtheorem{remark}[theorem]{Remark}

\newtheorem{example}[theorem]{Example}

\DeclareMathOperator{\rank}{rank}

\begin{document}
\title{The dimensions of Schur squares of HRS codes\footnote{The research was supported by the National Natural Science Foundation of China under the Grants 12222113 and 12441105.}}
\author{Haojie Gu\footnote{Haojie Gu is with the School of Mathematical Sciences, Capital Normal University, Beijing, China, 100048. Email: 2200502051@cnu.edu.cn.},
\and Zhihao Zhu\footnote{Zhihao Zhu is with the School of Mathematical Sciences, Capital Normal University, Beijing, China, 100048. Email: 2240501027@cnu.edu.cn.},
	\and Jun Zhang\footnote{Jun Zhang is with the School of Mathematical Sciences, Capital Normal University, Beijing, China, 100048. Email: junz@cnu.edu.cn.}
}
\date{}
\maketitle

\begin{abstract}
The Schur square of linear codes over a finite field has emerged as a fundamental operation in both classical and quantum coding theory. In this paper, we investigate the Schur square problem of Hyperderivative Reed-Solomon (HRS) codes. By solving certain special determinants, we first give a lower bound and an upper bound for the dimensions of Schur squares of HRS codes, and then prove that when $p\geq t\geq 2s$ and $t\leq \frac{r+2s-1}{2}$, the dimension of the Schur square of the HRS code $HRS_{t}(\{\alpha_{1},\dots,\alpha_{r}\},s)$ (with length $rs$ and dimension $t$) reaches the upper bound $(2t-2s+1)s$. In particular, when $p \ge t=2s$ and $r\geq t+1$, the dimension of the Schur square equals $\frac{t(t+1)}{2}$ which is the dimension of the Schur squares of random codes with high probability. As an application in code-based cryptography, HRS codes with specific parameter settings might resist the attack of Schur square distinguisher. 

% thereby generalizing the results on Schur squares of Reed-Solomon (RS) codes.
	\begin{flushleft}
		\textbf{Keywords: HRS codes, Schur squares, dimensions.} 
	\end{flushleft}
\end{abstract}

\section{Introduction}
Throughout this paper, let $\mathbb{F}_{q}$ be a finite field with size $q$ and characteristic $p$. Let $\mathbb{F}_{q}^n$ be the $n$-dimensional vector space over the finite field $\mathbb{F}_{q}$ and $Mat_{s\times r}(\mathbb{F}_{q})$ be the $\mathbb{F}_{q}$-vector space of $s\times r$ matrices with entries from $\mathbb{F}_{q}$. Vectors and matrices are respectively denoted in bold letters and bold capital letters such as $\boldsymbol{a}$ and $\boldsymbol{A}$. We always denote the entires of $\boldsymbol{A}\in Mat_{s\times r}(\mathbb{F}_{q})$ by $a_{i,j}$ for $i=1,\dots,s$ and $j=1,\dots r.$ 
%The Schur square $\boldsymbol{u}\star\boldsymbol{v}$ of vectors $\boldsymbol{u},\boldsymbol{v}\in \mathbb{F}_{q}^{n}$ is defined as:
%$$\boldsymbol{u}\star\boldsymbol{v}:=(u_{1}v_{1},\dots,u_{n}v_{n}).$$
%The $i$-th power $\boldsymbol{u}\star\cdots\star\boldsymbol{u}$ is denoted by $\boldsymbol{u}^{i}$.

%\subsection{Linear codes in the Hamming Metric}
For any vector $ \boldsymbol{x}=(x_1,x_2,\cdots,x_n)\in \mathbb{F}_{q}^n$, the \emph{Hamming weight} $\omega_{H}( \boldsymbol{x})$ of $ \boldsymbol{x}$ is defined to be the number of non-zero coordinates, i.e.,
$\omega_{H}( \boldsymbol{x})=|\left\{1\leq i\leq n\,|\, x_i\neq 0\right\}|.$

An $[n,k,d]_{q}$-linear code $\mathcal{C}\subseteq \mathbb{F}_{q}^n$ is a $k$-dimensional linear subspace of $\mathbb{F}_{q}^n$ with minimal distance $d=d_{H}(\mathcal{C})$ defined as $$d_{H}(\mathcal{C})=\min\left\{\omega_{H}(\boldsymbol{c}): \boldsymbol{c}\in\mathcal{C}\backslash\{0\}\right\}.$$
For any vector $\boldsymbol{u}\in\mathbb{F}_{q}^n$, the error distance from $\boldsymbol{u}$ to $\mathcal{C}$ is defined as
$$d_{H}(\boldsymbol{u},\mathcal{C})=\min\{d_{H}(\boldsymbol{u},\boldsymbol{v})\,|\,\boldsymbol{v}\in C\},$$
where $d_{H}(\boldsymbol{u},\boldsymbol{v})=|\{1\leq i\leq n\,|\,u_{i}\neq v_{i}\}|$
is the Hamming distance between vectors $\boldsymbol{u}$ and $\boldsymbol{v}$. The well-known Singleton bound says that $d_{H}\leq n-k+1$ for any $[n,k,d_{H}]$ linear code $\mathcal{C}$. If $d_{H}=n-k+1$, then $\mathcal{C}$ is called a maximum distance separable (MDS) code. The main ingredient in the construction of MDS codes is the generalized Reed-Solomon code, which itself is also MDS.

\begin{definition}
	Let $\boldsymbol{\alpha}=\left(\alpha_{1},\cdots,\alpha_{n}\right)\in \mathbb{F}_{q}^n$ with pairwise distinct $\alpha_{i}$'s, and $\boldsymbol{v}=\left(v_{1},\cdots,v_{n}\right)\in (\mathbb{F}_{q}^*)^n$. For $k\leq n$, the generalized Reed-Solomon (GRS) code $GRS_{\boldsymbol{v},k}(\mathcal{\boldsymbol{\alpha}})$ of length $n$, dimension $k$ and scaling vector $\boldsymbol{v}$ is defined as 
    \begin{equation*}
    GRS_{\boldsymbol{v},k}(\mathcal{\alpha})=\left\{(v_1f(\alpha_{1}),\cdots,v_nf(\alpha_{n})):f(x)\in \mathbb{F}_{q}[x],\, \deg(f)\leq k-1\right\}.
	\end{equation*}
\end{definition}

The Schur square of linear codes has emerged as a fundamental operation in both classical and quantum coding theory~\cite{randriambololona2015products}. Schur squares of Reed-Solomon, cyclic, Reed–Muller, hyperbolic, and toric codes have been studied and applied in both coding theory and cryptography.

\begin{definition}[Schur square~\cite{mirandola2012schur}]
   Let $\mathcal{C}$ be an $[n,k,d]_{q}$-linear code. The Schur square  code $\widehat{\mathcal{C}}$ of $\mathcal{C}$ is defined as:   
   $$\widehat{\mathcal{C}}:=\mathrm{Span}_{\mathbb{F}_{q}}\left\{\boldsymbol{c}_{1}\star\boldsymbol{c}_{2}:\boldsymbol{c}_{1},\boldsymbol{c}_{2}\in\mathcal{C}\right\}\subseteq \mathbb{F}_{q}^n,$$ 
   where $\boldsymbol{x}\star \boldsymbol{y}:= (x_1y_1,\cdots,x_ny_n)\in\mathbb{F}_{q}^n$  for $\boldsymbol{x}$= $(x_1,\cdots,x_{n}),\boldsymbol{y} = (y_1,\cdots,y_n)\in\mathbb{F}_{q}^n$. 
\end{definition}

Clearly, the Schur square $\widehat{\mathcal{C}}$ has the same length as the original code $\mathcal{C}$. Moreover, $\widehat{\mathcal{C}}$ contains a ``copy" of $\mathcal{C}$ under the map $\boldsymbol{c}\mapsto \boldsymbol{c}\star \boldsymbol{c}$, so we have the square distance $d_H(\widehat{\mathcal{C}})\leq d_H(\mathcal{C})$ and the square dimension $\widehat{k}=\dim_{\mathbb{F}_{q}}(\widehat{\mathcal{C}})\geq k$. Any $\mathbb{F}_{q}$-basis $\left\{\boldsymbol{g}_1,\cdots,\boldsymbol{g}_k\right\}$ of $\mathcal{C}$ gives a generator system $\left\{\boldsymbol{g}_{i}\star\boldsymbol{g}_{j}:1\leq i\leq j\leq k\right\}$ of $\mathcal{C}$ over $\mathbb{F}_{q}$, hence $\widehat{k}\leq\frac{k(k+1)}{2}$. Moreover, for a random $[n,k]$ linear code $\mathscr{C}$ over $\mathbb{F}_{q}$, the Schur square has dimension $\dim_{\mathbb{F}_{q}}(\widehat{\mathscr{C}})=\min\left\{n,\binom{k+1}{2}\right\}$ with high probability~\cite{randriambololona2015linear}.

The square dimension $\widehat{k}$ of $\mathcal{C}$ is an important parameter for the security of the McEliece cryptosystem based on $\mathcal{C}$. In the McEliece cryptosystem, one performs some random operation on the secret key, a generator matrix $G$ of $\mathcal{C}$, to hide the algebraic structure of $G$. By considering the efficiency of code-based cryptosystems, linear codes with strong algebraic structure would be preferable to reduce the key size. However, the strong algebraic structure often makes the square dimension significantly smaller than the corresponding value of a random code, which leads to a  distinguisher of the public code from random codes.  %However, the public key (a randomized generator matrix $G'$ from $G$) has significantly small square dimension compared with the random codes, which leads to a Schur square attack on the cryptosystem. 
Notable examples include GRS codes~\cite{marquez2013non}, low co-dimensional subcodes of GRS codes~\cite{wieschebrink2010cryptanalysis}, Reed-Muller codes~\cite{borodin2014effective}, Polar codes~\cite{druagoi2018vulnerabilities}, some Goppa codes~\cite{couvreur2016polynomial}, high rate alternant codes~\cite{faugere2013distinguisher} and algebraic geometry codes~\cite{couvreur2014polynomial,couvreur2017cryptanalysis}, among others. 

In this paper, we will study the Schur square of a variant of GRS codes, called Hyperderivative Reed-Solomon (HRS for short) code, and prove that they have large square dimensions.

%Clearly $\widehat{\mathcal{C}}$ has same length as $\mathcal{C}$. Moreover, $\widehat{\mathcal{C}}$ contains a ``copy" of $\mathcal{C}$, given by the  map $\boldsymbol{A}\mapsto\boldsymbol{A}\star\boldsymbol{A}$, hence it has dimension $\widehat{t}\geq t$ and distance $\widehat{d_{N}(\mathcal{C})}\leq d_{N}(\mathcal{C})$. Often we will call $\widehat{t}$ and $\widehat{d_{N}(\mathcal{C})}$ the product dimension and the product distance of $\mathcal{C}$.{\color{red} We will be interested in more precise estimates of $\widehat{t}$.}

%\subsection{Organization}
The organization is as follows. In Section 2, we introduce basic objects, including Hyperderivatives, HRS code and some useful determinants. In Section 3, we present a matrix whose row vectors generate the code \(\widehat{HRS_{t}(\boldsymbol{\alpha},s)}\) but not necessarily linearly independent. As a consequence, we obtain a lower bound and an upper bound for the square dimensions of HRS codes. We also show that the upper bound is achievable under certain conditions. In Section 4, we give a conclusion of this paper.%and we prove that when $p\geq t\geq 2s$ and $t\leq \frac{r-s+1}{2}$, the dimension of  $\widehat{HRS_{t}(\boldsymbol{\alpha},s)}$ is $(2t-2s+1)s$.

\section{Preliminaries}
In this section we first fix the notation used throughout the paper and then establish several lemmas on determinants of matrices with special structure that will be used in the sequel.

%\hspace{\fill}

%We always denote the entries of a vector $\mathbf{a}\in \mathbb{F}_{q}^{n}$ by $a_{1},\dots,a_{n}$.
\begin{itemize}
    \item Let $\boldsymbol{\alpha} = (\alpha_{1},\dots,\alpha_{r}) \in \mathbb{F}_{q}^{r}$.
For each integer $i \ge 0$, define
\[
  M_{i}(\boldsymbol{\alpha}) :=
  \begin{pmatrix}
    1          & \cdots & 1 \\
    \alpha_{1} & \cdots & \alpha_{r} \\
    \vdots     &        & \vdots \\
    \alpha_{1}^{i} & \cdots & \alpha_{r}^{i}
  \end{pmatrix}
  \in \mathbb{F}_{q}^{(i+1)\times r},
\]
and set
\[
  V(\boldsymbol{\alpha}):=\det(M_{r-1}(\boldsymbol{\alpha}) )
  = \prod_{1 \le s < t \le r} (\alpha_{t} - \alpha_{s}),
\]
the Vandermonde determinant associated with $\alpha_{1},\dots,\alpha_{r}$.
   \item For integers $a,b,i,j,m$ with $0 \le a \le b \le p-1$, $i \ge 0$ and $j \ge m-1$, define
    \[
      M(a,b,i,j,m)
      :=
      \begin{pmatrix}
        \binom{i}{a}\binom{j}{a}      & \cdots & \binom{i}{b}\binom{j}{b} \\
        \binom{i+1}{a}\binom{j-1}{a}  & \cdots & \binom{i+1}{b}\binom{j-1}{b} \\
        \vdots                        & \ddots & \vdots                     \\
        \binom{i+m-1}{a}\binom{j+1-m}{a} & \cdots & \binom{i+m-1}{b}\binom{j+1-m}{b}
      \end{pmatrix}
      \in \mathbb{F}_{q}^{m \times (b-a+1)}.
    \]
    When $m = b-a+1$, we abbreviate the notation to $M(a,b,i,j)$.

\item Under the same assumptions on $a,b,i,j,m$, and for integers $k$ and $\ell$ with
    $1 \le \ell \le b-a+1$, let $M_{\ell,k}(a,b,i,j,m)$ denote the matrix obtained from
    $M(a,b,i,j,m)$ by replacing its $\ell$-th column with the column vector
    \[
      \begin{pmatrix}
        \binom{i}{k}\binom{j}{k} \\
        \vdots \\
        \binom{i+m-1}{k}\binom{j+1-m}{k}
      \end{pmatrix}.
    \]
    In other words, $M_{\ell,k}(a,b,i,j,m)$ has the form
    \[
      \small{\begin{pmatrix}
        \binom{i}{a}\binom{j}{a}
        & \cdots
        & \binom{i}{a+\ell-2}\binom{j}{a+\ell-2}
        & \binom{i}{k}\binom{j}{k}
        & \binom{i}{a+\ell}\binom{j}{a+\ell}
        & \cdots
        & \binom{i}{b}\binom{j}{b}
        \\
        \vdots & \ddots & \vdots & \vdots & \vdots & \ddots & \vdots \\
        \binom{i+m-1}{a}\binom{j+1-m}{a}
        & \cdots
        & \binom{i+m-1}{a+\ell-2}\binom{j+1-m}{a+\ell-2}
        & \binom{i+m-1}{k}\binom{j+1-m}{k}
        & \binom{i+m-1}{a+\ell}\binom{j+1-m}{a+\ell}
        & \cdots
        & \binom{i+m-1}{b}\binom{j+1-m}{b}
      \end{pmatrix}.}
    \]
    When $m = b-a+1$, we abbreviate the notation to $M_{\ell,k}(a,b,i,j)$.

\end{itemize}
\subsection{ HRS codes}
The following notion is often referred to as the Hasse derivative. 
In this paper, however, we follow Skriganov \cite{skriganov2001coding} 
and use the term \emph{hyperderivative}, see also Lidl and Niederreiter \cite{lidl1997finite}.

\begin{definition}
Let $t$ be a positive integer, and let $f(x)=f_{0} + f_{1} x + \cdots + f_{t-1} x^{t-1} \in \mathbb{F}_{q}[x]_{\leq t-1}$.
For an integer $j$ with $0 \le j < p$, the $j$-th hyperderivative of $f(x)$ is the polynomial
\begin{equation*}
	\partial^{(j)} f(X)=\sum\limits_{i=0}^{t-1}\binom{i}{j}f_{i}X^{i-j},
	\end{equation*}
where we use the convention
$
  \binom{i}{j} =
  \begin{cases}
    \displaystyle \frac{i!}{(i-j)!\, j!}, & 0 \le j \le i,\\[0.8ex]
    0, & \text{otherwise}
  \end{cases}.
$
\end{definition}

\begin{remark}
    In fact, we can ignore the condition $j<p$: first compute the binomial coefficients over the integers $\mathbb{Z}$ and then consider the results by modulo $p$.
\end{remark}

Next, we give the definition of HRS codes.
\begin{definition}[\cite{can2025generalized}]
    Let $r,s,t$ be positive integers such that $1 \le s \le p$, $1 \le r \le q$ and 
$1 \le t \le rs$. Our construction is based on an $r$-tuple $\boldsymbol{\alpha} = (\alpha_{1},\dots,\alpha_{r}) \in \mathbb{F}_{q}^{r},
$
referred to as the list of evaluation points, and on an $s \times r$ matrix
\[
  \boldsymbol{V} = (v_{i,j})_{1 \le i \le s,\; 1 \le j \le r} \in \operatorname{Mat}_{s \times r}(\mathbb{F}_{q}^{*}),
\]
called the multiplier matrix. The HRS code $HRS_{t}(\boldsymbol{\alpha}, \boldsymbol{V}, s)$ is defined to be the image of the following evaluation map
\[
\begin{aligned}
  \operatorname{EV}_{\boldsymbol{\alpha},\boldsymbol{V}} :
    \mathbb{F}_{q}[x]_{\le t-1}
    &\longrightarrow \operatorname{Mat}_{s \times r}(\mathbb{F}_{q}),\\
  f(x) &\longmapsto
  \begin{pmatrix}
    v_{1,1}\,\partial^{(0)} f(\alpha_{1}) & \cdots & v_{1,r}\,\partial^{(0)} f(\alpha_{r}) \\
    v_{2,1}\,\partial^{(1)} f(\alpha_{1}) & \cdots & v_{2,r}\,\partial^{(1)} f(\alpha_{r}) \\
    \vdots                                & \ddots & \vdots                                \\
    v_{s,1}\,\partial^{(s-1)} f(\alpha_{1}) 
      & \cdots 
      & v_{s,r}\,\partial^{(s-1)} f(\alpha_{r})
  \end{pmatrix},
\end{aligned}
\]
where $\partial^{(i)} f$ denotes the $i$-th hyperderivative of $f$. When the multiplier matrix $\boldsymbol{V}$ is the all-ones matrix, we simply write $HRS_{t}(\boldsymbol{\alpha}, s)$ and $EV_{\boldsymbol{\alpha}}(f)$, where $f\in\mathbb{F}_{q}[x]_{\leq t-1}$.
\end{definition}

We fix a list of evaluation points $\boldsymbol{\alpha}=(\alpha_{1},\cdots,\alpha_{r})\in\mathbb{F}_{q}^r$. If the codewords $EV_{\boldsymbol{\alpha}}(f)$ of code $HRS_{t}(\boldsymbol{\alpha},s)$ are arranged from left to right and from top to bottom into an $sr$-dimensional vector 
\[\left(\partial^{(0)}(f(\boldsymbol{\alpha})),\partial^{(1)}(f(\boldsymbol{\alpha})),\cdots,\partial^{(s-1)}(f(\boldsymbol{\alpha}))\right),\]
then the generator matrix of $HRS_{t}(\boldsymbol{\alpha},s)$ becomes 
\begin{equation*}
   \begin{pmatrix}
    \boldsymbol{1}&\boldsymbol{0}&\cdots&\boldsymbol{0}\\
    \boldsymbol{\alpha}&\boldsymbol{1}&\cdots&\boldsymbol{0}\\
    \boldsymbol{\alpha}^2&\binom{2}{1}\boldsymbol{\alpha}&\cdots&\boldsymbol{0}\\
    \vdots&\vdots&\ddots&\vdots\\
    \boldsymbol{\alpha}^{t-1}&\binom{t-1}{1}\boldsymbol{\alpha}^{t-2}&\cdots&\binom{t-1}{s-1}\boldsymbol{\alpha}^{(t-1)-(s-1)}
  \end{pmatrix},
\end{equation*}
where $f\in\mathbb{F}_{q}[x]_{\leq t-1}$, $\partial^{(i)}(f(\boldsymbol{\alpha}))=(\partial^{(i)}f(\alpha_{1}),\cdots,\partial^{(i)}f(\alpha_{r}))$ for all $0\leq i\leq s-1$ and 
$$\binom{j}{i}\boldsymbol{\alpha}^{j-i}=\left(\binom{j}{i}{\alpha_{1}}^{j-i},\binom{j}{i}{\alpha_{2}}^{j-i},\cdots,\binom{j}{i}{\alpha_{r}}^{j-i}\right)\ \mbox{for all}\ 0\leq j\leq t-1,0\leq i\leq s-1.$$

The HRS codes can be viewed as a variant of the classic generalized Reed-Solomon codes. The concept was proposed by Rosenbloom and
Tsfasman in~\cite{ozen2006linear}.
The HRS codes achieve the Singleton-like bound under the Niederreiter-Rosenbloom-Tsfasman (NRT) metric which was first introduced by Niederreiter in~\cite{niederreiter1987point} and then by Rosenbloom and Tsfasman in~\cite{rosenbloom1997codes}. Gu and Zhang considered the decoding problem of HRS codes under NRT-metric and proposed a Welch-Berlekamp algorithm for the unique decoding in~\cite{gu2026unique} and a list decoding algorithm in~\cite{gu2026list}. The dimension of the Schur square of a linear code is a key indicator of Schur square distinguisher. We further observe an interesting phenomenon
 that the dimension of the Schur square of an HRS code could be quite large compared with that of a random code. Therefore, in this paper, we investigate the dimensions of Schur squares of HRS codes.

As we work on the matrices instead of vectors, we recall the definition of Schur squares on the space of matrices.
%Next, we extend codewords from vectors to matrices.
Denote by $\star$ the  Schur product on $\operatorname{Mat}_{s\times}(\mathbb{F}_{q})$, i.e., $\boldsymbol{A}\star\boldsymbol{B}:=(a_{i,j}b_{i,j})_{1\leq i\leq s\atop 1\leq j\leq r}\in\operatorname{Mat}_{s\times r}(\mathbb{F}_{q})$
where $\boldsymbol{A}=(a_{i,j})_{1\leq i\leq  s\atop 1\leq j\leq r},\boldsymbol{B}=(b_{i,j})_{1\leq i\leq s\atop 1\leq j\leq r}\in\operatorname{Mat}_{s\times r}(\mathbb{F}_{q})$.
\begin{definition}[Schur square code on $\operatorname{Mat}_{s\times r}(\mathbb{F}_{q})$]
   Let $\mathcal{C}\subseteq\operatorname{Mat}_{s\times r}(\mathbb{F}_{q})$ be an linear $[rs,t]$ code. We call $\widehat{\mathcal{C}}:=\operatorname{Span}_{\mathbb{F}_{q}}\left\{\boldsymbol{A}\star\boldsymbol{B}:\boldsymbol{A},\boldsymbol{B}\in\mathcal{C}\right\}\subseteq \operatorname{Mat}_{s\times r}(\mathbb{F}_{q})$ the Schur square code of $\mathcal{C}$. 
\end{definition}

\subsection{Some Determinants}
In this subsection, we compute several determinants that will be used in the proofs of our main theorems.

\begin{lemma}\label{Lem:2.1}
Let $a,b,i,j$ be integers with $0 \le a \le b \le p-1$, $i \ge 0$ and $j \ge b-a+1$. Then
\[
  \det(M(a,b,i,j))
  = \left(\prod_{\ell=0}^{b-a}
      \frac{\binom{i+\ell}{a}\binom{j-\ell}{a}}{\binom{a+\ell}{\ell}^{2}}\right)
    \cdot
    \left(\prod_{w=0}^{b-a-1}
      \frac{\displaystyle\prod_{u=1}^{b-a-w}\bigl(j - i + 1 - w - 2u\bigr)}
           {(b-a-w)!}\right).
\]
In particular, we have
\[
  \det(M(0,b,i,j))
  = \prod_{w=0}^{b-1}
      \frac{\displaystyle\prod_{u=1}^{b-w}\bigl(j - i + 1 - w - 2u\bigr)}
           {(b-w)!}.
\]
\end{lemma}

\begin{proof}
The proof is given in Appendix~A.
\end{proof}

\begin{lemma}\label{Lem:2.3}
Let $b,i,j,k,\ell$ be integers with $0 \le b \le p-1$, $i \ge 0$, $j \ge b+1$ and $ 1 \le \ell \le b+1 \le k$.
If $\ell = 1$, then
\small{\begin{equation*}
    \begin{aligned}
        \det(M_{1,k}(0,b,i,j)
  =& (-1)^{b}\,
    \frac{\displaystyle\prod_{h=0}^{b}(i+h)(j-h)}{(b!)^{2} \, k^{2}}
    \cdot
    \prod_{u=0}^{b-1}
      \frac{\displaystyle\prod_{h=0}^{b-1-u}\bigl(j-u-i-1-2h\bigr)}
           {(b-u-1)!\,(k-u-1)}
    \cdot
    \binom{i-1}{k-b-1}\binom{j-b-1}{k-b-1}.
    \end{aligned}
\end{equation*}}
If $\ell \ge 2$, then
\small{\begin{equation*}
\begin{aligned}
  \det\bigl(M_{\ell,k}(0,b,i,j)\bigr)
  &= (-1)^{\,b-\ell+1}\left(
     \prod_{m=0}^{\ell-2}
       \frac{
         \displaystyle\prod_{u=1}^{b-m}\bigl(j-m-i-2u+1\bigr)
       }{
         (b-m)! \,\dfrac{k-m}{\ell-m-1}
       }\right)
     \cdot\left(
     \frac{
       \displaystyle\prod_{h=0}^{b-\ell+1}(i+h)(j-\ell+1-h)
     }{
       \bigl((b-\ell+1)!\bigr)^{2}\,(k-\ell+1)^{2}
     } \right)\\
  &\quad\cdot\left(
     \prod_{u=0}^{b-\ell}
       \frac{
         \displaystyle\prod_{h=0}^{b-\ell-u}\bigl(j-u-i-\ell-2h\bigr)
       }{
         (b-\ell-u)!\,(k-\ell-u)
       }\right)
     \cdot
     \binom{i-1}{k-b-1}\binom{j-b-1}{k-b-1}.
\end{aligned}
\end{equation*}}
\end{lemma}

\begin{proof}
The proof is given in Appendix~A.
\end{proof}

To obtain an important result regarding the rank of a matrix, we first need to introduce a key determinant identity-Cauchy's determinant identity.
\begin{lemma}[Cauchy’s double alternant {\cite{cauchy1841memoire}}]
Let $r$ be a positive integer, and let $x_{1},\dots,x_{r}$ and
$y_{1},\dots,y_{r}$ be elements of a field such that
$x_{i} + y_{j} \neq 0$ for all $1 \le i,j \le r$. Then
\[
  \det\!\left(\frac{1}{x_{i}+y_{j}}\right)_{1 \le i,j \le r}
  =
  \frac{\displaystyle\prod_{1 \le i < j \le r}(x_{j}-x_{i})(y_{j}-y_{i})}%
       {\displaystyle\prod_{1 \le i,j \le r}(x_{i}+y_{j})}.
\]
\end{lemma}

\begin{lemma}\label{Lem:2.2}
Let $p,t,s,i$ be positive integers such that $2s \le t\leq q$ and
$
  2t - 2s + 1 \;\le\; i \;\le\; 2t - 2.
$
Define $A= t - \Bigl\lceil \frac{i}{2} \Bigr\rceil$ and 
  $ B= s - t + \Bigl\lceil \frac{i}{2} \Bigr\rceil$,
so that $1 \le A,B \le s-1$ and $A + B = s$.
Consider the matrices
\[
  M\bigl(0,A-1,i-t+1,t-1\bigr) \in \mathbb{F}_q^{A\times A}
  \quad\text{and}\quad
  M\bigl(A,s-1,i-t+1,t-1,A\bigr) \in \mathbb{F}_q^{A\times B},
\]
and let 
\[
  M\bigl(0,A-1,i-t+1,t-1\bigr)\begin{pmatrix}
    c_{1,1} & \cdots & c_{1,B} \\
    \vdots  &        & \vdots  \\
    c_{A,1} & \cdots & c_{A,B}
  \end{pmatrix}
  = 
  M\bigl(A,s-1,i-t+1,t-1,A\bigr).
\]
Then the following statements hold:
\begin{enumerate}
  \item[(i)] If $2t - 2s + 1 \le i \le 2t - s - 1$ (in which case one has $A \ge B$),
  then the $B\times B$ leading principal submatrix
  \[
    \bigl(c_{\ell,j}\bigr)_{1 \le \ell \le B,\; 1 \le j \le B}
  \]
  is invertible.

  \item[(ii)] If $2t - s \le i \le 2t - 2$ (in which case one has $A \le B$),
  then the $A\times A$ submatrix
  \[
    \bigl(c_{\ell,j}\bigr)_{1 \le \ell \le A,\; B - A + 1 \le j \le B}
  \]
  is invertible.
\end{enumerate}
\end{lemma}

\begin{proof}
Please refer to Appendix~A for the proof.
\end{proof}

\section{The dimensions of Schur squares of HRS codes}
In this section, we investigate the dimensions of Schur squares of HRS codes. First, we give a generating system for the Schur square codes. Then we present our main result on the square dimensions of HRS codes. Finally, we list some numerical results to support our main result.

Thereinafter, we assume that $s\geq 2.$ For $s=1$, the HRS code is reduced to the classic GRS code $GRS_{\boldsymbol{v},t}(\boldsymbol{\alpha})$ with lenght $r$ and dimension $t$, whose square dimension equals $\min\{2t-1, r\}$~\cite{mirandola2012schur}.
%\begin{proposition}[{\cite{mirandola2012schur}}]
 %   Let $\mathcal{C}$ be an $[n,k]$ MDS code. If $2k-1\leq n$, then $\widehat{k}\geq 2k-1$. In particular, if $2k-1\leq n$, then the dimension of Schur square code of Reed-Solomon code $RS_{k}(\boldsymbol{\alpha})$ is $2k-1$.
%\end{proposition}
We first give a generating system for the Schur square code.
\begin{theorem}\label{thm:Schur-square-structure}
Let $q = p^{m}$ be a prime power with characteristic $p\geq 2s$.
Let $r,s,t$ be integers satisfying $s\geq 2,t\geq 2s, 1 \le t \le rs-1$ and $r\ge 2t-2s+1$. For $0\leq i\leq 2t-2$, let $A=t - \Bigl\lceil \frac{i}{2} \Bigr\rceil,B= s - t + \Bigl\lceil \frac{i}{2} \Bigr\rceil$.
%\[
%  t \ge 2s,\qquad 1 \le t \le rs-1,\qquad 
%\]
Let $\boldsymbol{\alpha} = (\alpha_{1},\dots,\alpha_{r}) \in \mathbb{F}_{q}^{r}$ 
with pairwise distinct entries, and let
  $\mathcal{C} = HRS_{t}(\boldsymbol{\alpha},s)$
be the hyperderivative Reed–Solomon code of length $rs$ and dimension $t$.   
 Then the Schur square code $\widehat{\mathcal{C}}$ of code $\mathcal{C}$ is generated by the rows of the following matrix:
\begin{equation}\label{matrix:gen}
  \begin{pmatrix}
    M_{2t-2s}(\boldsymbol{\alpha}) &&&&& \\[2pt]
    & M_{2t-2s-2}(\boldsymbol{\alpha}) &&&& \\
    && \ddots &&& \\
    &&& M_{2t-4s+6}(\boldsymbol{\alpha}) && \\
    &&&& M_{2t-4s+4}(\boldsymbol{\alpha}) & \\
    &&&&& M_{2t-4s+2}(\boldsymbol{\alpha}) \\[4pt]
    \boldsymbol{\alpha}^{2t-2s+1} & \boldsymbol{0} & \cdots & \boldsymbol{0} &
    \boldsymbol{0} & c_{1,1}^{(1)} \boldsymbol{\alpha}^{2t-4s+3} \\
    \vdots & \vdots & \vdots & \vdots & \vdots & \vdots \\
    \boldsymbol{0} & \boldsymbol{0} & \cdots & \boldsymbol{0} &
    \boldsymbol{\alpha}^{2t-4s+5} & c_{s-1,1}^{(1)} \boldsymbol{\alpha}^{2t-4s+3} \\[4pt]
    \boldsymbol{\alpha}^{2t-2s+2} & \boldsymbol{0} & \cdots & \boldsymbol{0} &
    \boldsymbol{0} & c_{1,1}^{(2)} \boldsymbol{\alpha}^{2t-4s+4} \\
    \vdots & \vdots & \vdots & \vdots & \vdots & \vdots \\
    \boldsymbol{0} & \boldsymbol{0} & \cdots & \boldsymbol{0} &
    \boldsymbol{\alpha}^{2t-4s+6} & c_{s-1,1}^{(2)} \boldsymbol{\alpha}^{2t-4s+4} \\[4pt]
    \boldsymbol{\alpha}^{2t-2s+3} & \boldsymbol{0} & \cdots & \boldsymbol{0} &
    c_{1,1}^{(3)} \boldsymbol{\alpha}^{2t-4s+7} &
    c_{1,2}^{(3)} \boldsymbol{\alpha}^{2t-4s+5} \\
    \vdots & \vdots & \vdots & \vdots & \vdots & \vdots \\
    \boldsymbol{0} & \boldsymbol{0} & \cdots &
    \boldsymbol{\alpha}^{2t-4s+9} &
    c_{s-2,1}^{(3)} \boldsymbol{\alpha}^{2t-4s+7} &
    c_{s-2,2}^{(3)} \boldsymbol{\alpha}^{2t-4s+5} \\
    \vdots & \vdots & \vdots & \vdots & \vdots & \vdots \\[4pt]
    \boldsymbol{\alpha}^{2t-3} &
    c_{1,1}^{(2s-3)} \boldsymbol{\alpha}^{2t-5} & \cdots &
    c_{1,s-3}^{(2s-3)} \boldsymbol{\alpha}^{2t-2s+3} &
    c_{1,s-2}^{(2s-3)} \boldsymbol{\alpha}^{2t-2s+1} &
    c_{1,s-1}^{(2s-3)} \boldsymbol{\alpha}^{2t-2s-1} \\[2pt]
    \boldsymbol{\alpha}^{2t-2} &
    c_{1,1}^{(2s-2)} \boldsymbol{\alpha}^{2t-4} & \cdots &
    c_{1,s-3}^{(2s-2)} \boldsymbol{\alpha}^{2t-2s+4} &
    c_{1,s-2}^{(2s-2)} \boldsymbol{\alpha}^{2t-2s+2} &
    c_{1,s-1}^{(2s-2)} \boldsymbol{\alpha}^{2t-2s}
  \end{pmatrix},
\end{equation}
where the
coefficients $c_{\ell,j}^{(i-2t+2s)} \in \mathbb{F}_{q}$ are uniquely determined
 by the relation
\[
  M\bigl(0,A-1,i-t+1,t-1\bigr)
  \begin{pmatrix}
    c_{1,1}^{(i-2t+2s)} & \cdots & c_{1,B}^{(i-2t+2s)} \\
    \vdots & \ddots & \vdots \\
    c_{A,1}^{(i-2t+2s)} & \cdots & c_{A,B}^{(i-2t+2s)}
  \end{pmatrix}
  =
  M\bigl(A,s-1,i-t+1,t-1,A\bigr)
\]
for each  $2t-2s+1 \le i \le 2t-2,1\leq\ell\leq A,1\leq j\leq B$. 
\end{theorem}

\begin{proof}
As recalled earlier, the Schur square code $\widehat{\mathcal{C}}$ is generated by
\begin{equation*}
  \Bigl(
    \boldsymbol{\alpha}^{j_{1}+j_{2}},
    \binom{j_{1}}{1}\binom{j_{2}}{1}\boldsymbol{\alpha}^{j_{1}+j_{2}-2},
    \dots,
    \binom{j_{1}}{s-1}\binom{j_{2}}{s-1}\boldsymbol{\alpha}^{j_{1}+j_{2}-2(s-1)}
  \Bigr)_{0 \le j_{1} \le j_{2} \le t-1},
\end{equation*}
where we use the convention that $\binom{j_{\ell}}{k}=0$ if $j_{\ell}<k$,
for $\ell=1,2$ and $0 \le k \le s-1$.

For a fixed integer $i$ with $0 \le i \le 2t-2$, we collect all rows whose
first block is $\boldsymbol{\alpha}^{i}$, i.e.\ all pairs $(j_{1},j_{2})$ with
$j_{1}+j_{2}=i$ and $0 \le j_{1} \le j_{2} \le t-1$.  After reordering these
rows, the submatrix corresponding to the first block
$\boldsymbol{\alpha}^{i}$ can be written as
\begin{equation}\label{equ:gen-alpha-i}
  \begin{pmatrix}
    \boldsymbol{\alpha}^{i} &
    \binom{0}{1}\binom{i}{1}\boldsymbol{\alpha}^{i-2} & \cdots &
    \binom{0}{s-1}\binom{i}{s-1}\boldsymbol{\alpha}^{i-2(s-1)} \\[2pt]
    \boldsymbol{\alpha}^{i} &
    \binom{1}{1}\binom{i-1}{1}\boldsymbol{\alpha}^{i-2} & \cdots &
    \binom{1}{s-1}\binom{i-1}{s-1}\boldsymbol{\alpha}^{i-2(s-1)} \\
    \vdots & \vdots & \ddots & \vdots \\
    \boldsymbol{\alpha}^{i} &
    \binom{\lfloor i/2 \rfloor}{1}\binom{\lceil i/2 \rceil}{1}\boldsymbol{\alpha}^{i-2}
    & \cdots &
    \binom{\lfloor i/2 \rfloor}{s-1}\binom{\lceil i/2 \rceil}{s-1}
    \boldsymbol{\alpha}^{i-2(s-1)}
  \end{pmatrix},
\end{equation}
where $\binom{j}{\ell}\binom{i-j}{\ell}\boldsymbol{\alpha}^{i-2\ell}=\boldsymbol{0}$ if $\min\{j,i-j\}<\ell$ or $i<2\ell$ for all $0\le\ell\le s-1$ and $0\le j\le\lfloor\frac{i}{2}\rfloor$.

Next, we divide $i\in [0, 2t-2]$ into four ranges.

\medskip
\noindent\textbf{Case 1: $0 \le i \le 2s-2 < t$.}
In this range the binomial coefficients vanish for $s>\lfloor i/2\rfloor$, so
matrix~\eqref{equ:gen-alpha-i} has the form
\begin{equation*}
  \begin{pmatrix}
    \boldsymbol{\alpha}^{i} & \boldsymbol{0} & \cdots & \boldsymbol{0}
                              & \cdots & \boldsymbol{0} \\
    \boldsymbol{\alpha}^{i} &
    \binom{1}{1}\binom{i-1}{1}\boldsymbol{\alpha}^{i-2} & \cdots &
    \boldsymbol{0} &  \cdots & \boldsymbol{0} \\
    \vdots & \vdots & \ddots & \vdots  & \ddots & \vdots \\
    \boldsymbol{\alpha}^{i} &
    \binom{\lfloor i/2 \rfloor}{1}\binom{\lceil i/2 \rceil}{1}\boldsymbol{\alpha}^{i-2}
    & \cdots &
    \binom{\lfloor i/2 \rfloor}{\lfloor i/2 \rfloor}
    \binom{\lceil i/2 \rceil}{\lfloor i/2 \rfloor}
    \boldsymbol{\alpha}^{i-2\lfloor i/2 \rfloor}  & \cdots & \boldsymbol{0}
  \end{pmatrix}.
\end{equation*}
The $(\lfloor i/2 \rfloor+1)\times(\lfloor i/2 \rfloor+1)$ submatrix of
scalar coefficients is exactly
\[
  M\bigl(0,\lfloor i/2 \rfloor,0,i\bigr),
\]
whose determinant is
\[
  \det\bigl(M(0,\lfloor i/2 \rfloor,0,i)\bigr)
  = \prod_{j=0}^{\lfloor i/2 \rfloor}
      \binom{i-j}{j}.
\]
Since $i\leq 2s-2< p$, all these binomial coefficients are
nonzero in $\mathbb{F}_{p}$, hence the determinant is nonzero and the matrix
is invertible.  Therefore matrix~\eqref{equ:gen-alpha-i} is row–equivalent to
\begin{equation}\label{equ:case-i-0-2s-2}
  \begin{pmatrix}
    \boldsymbol{\alpha}^{i} & \boldsymbol{0} & \cdots & \boldsymbol{0}
                            & \boldsymbol{0} & \cdots & \boldsymbol{0} \\
    \boldsymbol{0} & \boldsymbol{\alpha}^{i-2} & \cdots & \boldsymbol{0}
                   & \boldsymbol{0} & \cdots & \boldsymbol{0} \\
    \vdots & \vdots & \ddots & \vdots & \vdots & \ddots & \vdots \\
    \boldsymbol{0} & \boldsymbol{0} & \cdots
                   & \boldsymbol{\alpha}^{i-2\lfloor i/2 \rfloor}
                   & \boldsymbol{0} & \cdots & \boldsymbol{0}
  \end{pmatrix}.
\end{equation}

\medskip
\noindent\textbf{Case 2: $2s-1 \le i \le t-1$.}
In this range we obtain
\begin{equation*}
  \begin{pmatrix}
    \boldsymbol{\alpha}^{i} & \boldsymbol{0} & \cdots & \boldsymbol{0} \\
    \boldsymbol{\alpha}^{i} &
    \binom{1}{1}\binom{i-1}{1}\boldsymbol{\alpha}^{i-2} & \cdots &
    \boldsymbol{0} \\
    \vdots & \vdots & \ddots & \vdots \\
    \boldsymbol{\alpha}^{i} &
    \binom{s-1}{1}\binom{i-s+1}{1}\boldsymbol{\alpha}^{i-2} & \cdots &
    \binom{s-1}{s-1}\binom{i-s+1}{s-1}\boldsymbol{\alpha}^{i-2(s-1)} \\
    \vdots & \vdots & \ddots & \vdots \\
    \boldsymbol{\alpha}^{i} &
    \binom{\lfloor i/2 \rfloor}{1}\binom{\lceil i/2 \rceil}{1}\boldsymbol{\alpha}^{i-2}
    & \cdots &
    \binom{\lfloor i/2 \rfloor}{s-1}\binom{\lceil i/2 \rceil}{s-1}\boldsymbol{\alpha}^{i-2(s-1)}
  \end{pmatrix}.
\end{equation*}
The $s\times s$ submatrix of coefficients in the last $s$ rows is 
\[
  M\Bigl(0,s-1,\lfloor i/2 \rfloor-s+1,\lceil i/2 \rceil + s-1\Bigr).
\]
By Lemma~\ref{Lem:2.1} its determinant equals
\[
  \prod_{w=0}^{s-2}
    \frac{\displaystyle\prod_{u=1}^{s-1-w}
          \bigl(2s+\lceil i/2 \rceil - \lfloor i/2 \rfloor -1 -w-2u\bigr)}
         {(s-1-w)!},
\]
and all factors lie in $\{1,\dots,p-1\}$ because of $p>2s-2$.
Hence the determinant is nonzero and matrix~\eqref{equ:gen-alpha-i} is
row–equivalent to
\begin{equation}\label{equ:case-i-2s-1--t-1}
  \begin{pmatrix}
    \boldsymbol{\alpha}^{i} & \boldsymbol{0} & \cdots & \boldsymbol{0} \\
    \boldsymbol{0} & \boldsymbol{\alpha}^{i-2} & \cdots & \boldsymbol{0} \\
    \vdots & \vdots & \ddots & \vdots \\
    \boldsymbol{0} & \boldsymbol{0} & \cdots &
    \boldsymbol{\alpha}^{i-2(s-1)}
  \end{pmatrix}.
\end{equation}

\medskip
\noindent\textbf{Case 3: $t \le i \le 2t-2s$.}
In this range the same pattern as in Case~2 occurs, but with shifted indices
$\bigl(i-t+1,t-1\bigr)$ instead of
$\bigl(\lfloor i/2\rfloor-s+1,\lceil i/2\rceil+s-1\bigr)$.  By applying
Lemma~\ref{Lem:2.1} to the corresponding $s\times s$ submatrix, we again obtain
a nonzero determinant, and matrix~\eqref{equ:gen-alpha-i} is row–equivalent to
\begin{equation}\label{equ:case-i-t--2t-2s}
  \begin{pmatrix}
    \boldsymbol{\alpha}^{i} & \boldsymbol{0} & \cdots & \boldsymbol{0} \\
    \boldsymbol{0} & \boldsymbol{\alpha}^{i-2} & \cdots & \boldsymbol{0} \\
    \vdots & \vdots & \ddots & \vdots \\
    \boldsymbol{0} & \boldsymbol{0} & \cdots &
    \boldsymbol{\alpha}^{i-2(s-1)}
  \end{pmatrix}.
\end{equation}
(We omit the details, as they are completely analogous to Case~2.)

\medskip
\noindent\textbf{Case 4: $2t-2s+1 \le i \le 2t-2$.}
In this final range we have $A = t-\lceil i/2\rceil < s$ and
$1 \le B \le s-1$.  The submatrix~\eqref{equ:gen-alpha-i} can be rewritten as
\begin{equation*}
  \begin{pmatrix}
    \boldsymbol{\alpha}^{i} &
    \binom{i-t+1}{1}\binom{t-1}{1}\boldsymbol{\alpha}^{i-2} & \cdots &
    \binom{i-t+1}{s-1}\binom{t-1}{s-1}\boldsymbol{\alpha}^{i-2(s-1)} \\
    \boldsymbol{\alpha}^{i} &
    \binom{i-t+2}{1}\binom{t-2}{1}\boldsymbol{\alpha}^{i-2} & \cdots &
    \binom{i-t+2}{s-1}\binom{t-2}{s-1}\boldsymbol{\alpha}^{i-2(s-1)} \\
    \vdots & \vdots & \ddots & \vdots \\
    \boldsymbol{\alpha}^{i} &
    \binom{\lfloor i/2 \rfloor}{1}\binom{\lceil i/2 \rceil}{1}\boldsymbol{\alpha}^{i-2}
    & \cdots &
    \binom{\lfloor i/2 \rfloor}{s-1}\binom{\lceil i/2 \rceil}{s-1}\boldsymbol{\alpha}^{i-2(s-1)}
  \end{pmatrix}.
\end{equation*}
The $A\times A$ submatrix of coefficients in the first $A$ rows and first $A$
nontrivial columns is
\[
  M\bigl(0,A-1,i-t+1,t-1\bigr).
\]
From Lemma~\ref{Lem:2.1}, we have

\[
\det(  M\bigl(0,A-1,i-t+1,t-1\bigr))= \prod_{w=0}^{A-2}
    \frac{\displaystyle\prod_{u=1}^{A-1-w}\bigl(2t-i-1-w-2u\bigr)}
         {(A-1-w)!}.
\]
Since
\[
  1
  \le 2\Bigl\lceil \frac{i}{2} \Bigr\rceil - i + 1
  = 2t-i-1-2(A-1)
  \le 2t-i-1-w-2u
  \le 2t-i-3
  \le 2s-4 < p
\]
for all $0\le w\le A-2$ and $1\le u\le A-1-w$, we obtain 
$\det\bigl(M(0,A-1,i-t+1,t-1)\bigr)\neq 0$ in $\mathbb{F}_{p}$.
Therefore, matrix~\eqref{equ:gen-alpha-i} is row–equivalent to
\begin{equation}\label{equ:case-i-2t-2s+1--2t-2}
  \begin{pmatrix}
    \boldsymbol{\alpha}^{i} & \boldsymbol{0} & \cdots & \boldsymbol{0} &
    c_{1,1}^{(i-2t+2s)}\boldsymbol{\alpha}^{i-2A} & \cdots &
    c_{1,B}^{(i-2t+2s)}\boldsymbol{\alpha}^{i-2(s-1)} \\
    \boldsymbol{0} & \boldsymbol{\alpha}^{i-2} & \cdots & \boldsymbol{0} &
    c_{2,1}^{(i-2t+2s)}\boldsymbol{\alpha}^{i-2A} & \cdots &
    c_{2,B}^{(i-2t+2s)}\boldsymbol{\alpha}^{i-2(s-1)} \\
    \vdots & \vdots & \ddots & \vdots & \vdots & \ddots & \vdots \\
    \boldsymbol{0} & \boldsymbol{0} & \cdots &
    \boldsymbol{\alpha}^{i-2(A-1)} &
    c_{A,1}^{(i-2t+2s)}\boldsymbol{\alpha}^{i-2A} & \cdots &
    c_{A,B}^{(i-2t+2s)}\boldsymbol{\alpha}^{i-2(s-1)}
  \end{pmatrix},
\end{equation}
where the coefficients $c_{\ell,j}^{(i-2t+2s)}$ are uniquely determined by
\[
  M\bigl(0,A-1,i-t+1,t-1\bigr)
  \begin{pmatrix}
    c_{1,1}^{(i-2t+2s)} & \cdots & c_{1,B}^{(i-2t+2s)} \\
    \vdots & \ddots & \vdots \\
    c_{A,1}^{(i-2t+2s)} & \cdots & c_{A,B}^{(i-2t+2s)}
  \end{pmatrix}
  =
  M\bigl(A,s-1,i-t+1,t-1,A\bigr).
\]

\medskip
Collecting the row–equivalent forms obtained in
\eqref{equ:case-i-0-2s-2},
\eqref{equ:case-i-2s-1--t-1},
\eqref{equ:case-i-t--2t-2s},
and \eqref{equ:case-i-2t-2s+1--2t-2} for all $0 \le i \le 2t-2$, and ordering
the rows by increasing degree of the first block, we obtain a matrix whose row vectors generate the code $\widehat{\mathcal{C}}$ and which has exactly the block form
displayed in~\eqref{matrix:gen}. This completes the proof.
\end{proof}
%\begin{remark}
 %   When $s=1$, there are mainly Cases $2$ and $3$. Consequently, the Schur square code $\widehat{\mathcal{C}}$ of code $\mathcal{C}$ has a generator matrix $M_{2t-2}(\boldsymbol{\alpha})$. Hence, the dimension of $\widehat{\mathcal{C}}$ is $2t-1$. {\color{red} This coincides ... So we will focus on the case $s\geq 2$.}
%\end{remark}

Based on the above generating system, we present the main result on the dimensions of Schur squares of HRS codes.
\begin{theorem}\label{thm:dim-Schur-square-HRS}
Let $q = p^{m}$ with $p$ a prime, and let
$r,s,t$ be integers such that
  {$s\geq 2, 2s \le t \le rs - 1$ }
  and  $r \ge 2t - 2s + 1$.  Let $\boldsymbol{\alpha} = (\alpha_{1},\dots,\alpha_{r}) \in \mathbb{F}_{q}^{r}$ 
with pairwise distinct entries, and let
  $\mathcal{C} = HRS_{t}(\boldsymbol{\alpha},s)$
be the HRS code of length $rs$ and dimension $t$. Let $\widehat{\mathcal{C}}$ be the Schur square code of $\mathcal{C}$. Then the followings hold:
\begin{enumerate}
%  \item[$(i)$] If $t<2s$ and $p\geq t$, then then the Schur square code $\widehat{\mathcal{C}}$ of code $\mathcal{C}$
 %  has dimension
  % $$\dim(\widehat{\mathcal{C}})=t(t+1)/2.$$
  
  \item [$(i)$]
  If $p \ge 2s$, then
  \[
    (2t - 3s + 2)s
    \;\le\;
    \dim(\widehat{\mathcal{C}})\leq
    (2t - 2s + 1)s.
  \]

  \item[$(ii)$] If $p \ge t$, then
  \[
    \dim(\widehat{\mathcal{C}})=
     (2t - 2s + 1)\,s.
  \]
\end{enumerate}
In particular, when $p \ge t=2s$, we have
\[
  \dim(\widehat{\mathcal{C}})
  = \frac{t(t+1)}{2}.
\]
\end{theorem}
%\begin{remark}
%In fact, analogous to the Schur square of RS code $RS_{t}(\left\{\alpha_{1},\cdots,\alpha_{r}\right\},s)$ having dimension $\min\left\{r,2t-1\right\}$, the Schur product of HRS code $HRS_{t}({\alpha_{1},\cdots,\alpha_{r}},s)$ has dimension $\min\left\{rs,(2t-2s+1)s\right\}$. 
%\end{remark}

\begin{proof}%[Proof of Theorem~\ref{thm:dim-Schur-square-HRS}] 
By Theorem~\ref{thm:Schur-square-structure}, the Schur square code $\widehat{\mathcal{C}}$ is generated by the~(\ref{matrix:gen}).
\begin{enumerate}
    \item[(i)] We first consider the upper block-diagonal part $$\text{diag}(M_{2t-2s}(\boldsymbol{\alpha}),M_{2t-2s-2}(\boldsymbol{\alpha}),\cdots,M_{2t-4s+2}(\boldsymbol{\alpha})).$$Since the entries of \(\boldsymbol{\alpha}\) are pairwise distinct and $r \ge 2t-2s+1 \ge 2t-2s-2j+1$, each block \(M_{2t-2s-2j}(\boldsymbol{\alpha})\) has full row rank, i.e.
\[
\rank(M_{2t-2s-2j}(\boldsymbol{\alpha})) = 2t-2s-2j+1,\quad j=0,\cdots,s-1.
\]
As these blocks lie in disjoint column blocks, so
\[
(2t-3s+2)s=
\sum_{j=0}^{s-1}(2t-2s-2j+1)\le \dim(\widehat{\mathcal{C}})
.
\]
On the other hand, since the number of rows of the matrix in \eqref{matrix:gen} is
\[
(2t-3s+2)s+2\sum_{i=1}^{s-1}i=s(2t-2s+1),
\]
we obtain $\dim(\widehat{\mathcal{C}})\le s(2t-2s+1)$. Therefore,
\[
(2t-3s+2)s \le\dim(\widehat{\mathcal{C}})\le (2t-2s+1)s.
\]
    \item[(ii)] For $p\ge t$, it follows from the proof of (i) that it suffices to show that the rows of the matrix in \eqref{matrix:gen} are linearly independent. According to the proof of Theorem~\ref{thm:Schur-square-structure}, for each $2t-2s+1\le i \le 2t-2$, the corresponding submatrix is of the form \begin{equation*}
  \begin{pmatrix}
    \boldsymbol{\alpha}^{i} & \boldsymbol{0} & \cdots & \boldsymbol{0} &
    c_{1,1}^{(i-2t+2s)}\boldsymbol{\alpha}^{i-2A} & \cdots &
    c_{1,B}^{(i-2t+2s)}\boldsymbol{\alpha}^{i-2(s-1)} \\
    \boldsymbol{0} & \boldsymbol{\alpha}^{i-2} & \cdots & \boldsymbol{0} &
    c_{2,1}^{(i-2t+2s)}\boldsymbol{\alpha}^{i-2A} & \cdots &
    c_{2,B}^{(i-2t+2s)}\boldsymbol{\alpha}^{i-2(s-1)} \\
    \vdots & \vdots & \ddots & \vdots & \vdots & \ddots & \vdots \\
    \boldsymbol{0} & \boldsymbol{0} & \cdots &
    \boldsymbol{\alpha}^{i-2(A-1)} &
    c_{A,1}^{(i-2t+2s)}\boldsymbol{\alpha}^{i-2A} & \cdots &
    c_{A,B}^{(i-2t+2s)}\boldsymbol{\alpha}^{i-2(s-1)}
  \end{pmatrix}
\end{equation*}
where $A = t-\lceil i/2\rceil$ and $B=s-t+\lceil i/2\rceil$. Moreover, Lemma~\ref{Lem:2.2} implies that:
    \begin{enumerate}
        \item For $2t-2s+1\le i \le 2t-s-1$,  Lemma~\ref{Lem:2.2}(i) shows that the \(B\times B\) submatrix
  \[
    \bigl(c_{\ell,j}^{(i-2t+2s)}\bigr)_{1\le \ell,j\le B}
  \]
  is invertible. Therefore, in this case, \eqref{equ:case-i-2t-2s+1--2t-2} is row-equivalent to
     \small{   \begin{equation*}
  \begin{pmatrix}
    \boldsymbol{\alpha}^{i} & \boldsymbol{0} & \cdots & \boldsymbol{0} & \boldsymbol{0}
     &  \cdots&  \boldsymbol{0} & c_{1,1}^{(i-2t+2s)}\boldsymbol{\alpha}^{i-2A} & \cdots &
    c_{1,B}^{(i-2t+2s)}\boldsymbol{\alpha}^{i-2(s-1)}\\
    \boldsymbol{0} & \boldsymbol{\alpha}^{i-2} & \cdots & \boldsymbol{0} & \boldsymbol{0}
     & \cdots & \boldsymbol{0} & c_{2,1}^{(i-2t+2s)}\boldsymbol{\alpha}^{i-2A} & \cdots &
    c_{2,B}^{(i-2t+2s)}\boldsymbol{\alpha}^{i-2(s-1)}\\
    \vdots & \vdots & \ddots & \vdots & \vdots & \ddots & \vdots  &\vdots& &\vdots\\
    \boldsymbol{0} & \boldsymbol{0} & \cdots &
    \boldsymbol{\alpha}^{i-2(B-1)} &
   \boldsymbol{0} &\cdots & \boldsymbol{0} & c_{B,1}^{(i-2t+2s)}\boldsymbol{\alpha}^{i-2A} & \cdots &
    c_{B,B}^{(i-2t+2s)}\boldsymbol{\alpha}^{i-2(s-1)} \\
    * & * & \cdots &
    * &  \boldsymbol{\alpha}^{i-2B}
    & \cdots & \boldsymbol{0} & \boldsymbol{0} & \cdots &
    \boldsymbol{0} \\
    \vdots & \vdots & \ddots &
    \vdots & \vdots
    & & \vdots &\vdots &  &
    \vdots \\
    * & * & \cdots &
    * & *
    & \cdots & \boldsymbol{\alpha}^{i-2(A-1)} & \boldsymbol{0} & \cdots &
    \boldsymbol{0} 
  \end{pmatrix},
\end{equation*}}
where each $*$ denotes a scalar multiple of the power vector appearing in the same column. Moreover, \eqref{equ:case-i-2t-2s+1--2t-2} is row-equivalent to a matrix of the form
 \begin{equation*}
 \small{ \begin{pmatrix}
     * & \cdots & * & \boldsymbol{0}
     &  \cdots&  \boldsymbol{0} & \boldsymbol{\alpha}^{i-2A} &\boldsymbol{0} & \cdots &
    \boldsymbol{0}\\
    * & \cdots & * & \boldsymbol{0}
     & \cdots & \boldsymbol{0} & \boldsymbol{0} &  \boldsymbol{\alpha}^{i-2(A+1)} &\cdots &
    \boldsymbol{0}\\
    \vdots & \ddots & \vdots & \vdots & \ddots & \vdots & \vdots  &\vdots& \ddots &\vdots\\
    * & \cdots &
   * &
   \boldsymbol{0} &\cdots & \boldsymbol{0} & \boldsymbol{0} & \boldsymbol{0} & \cdots &
    \boldsymbol{\alpha}^{i-2(s-1)} \\
     * & \cdots &
    * &  \boldsymbol{\alpha}^{i-2B}
    & \cdots & \boldsymbol{0} & \boldsymbol{0} &\boldsymbol{0} & \cdots &
    \boldsymbol{0} \\
    \vdots & \ddots & \vdots &
    \vdots & \ddots
    & \vdots & \vdots &\vdots &  &
    \vdots \\
    * & \cdots &
    * & *
    & \cdots & \boldsymbol{\alpha}^{i-2(A-1)} & \boldsymbol{0} & \boldsymbol{0} & \cdots &
    \boldsymbol{0} 
  \end{pmatrix},}
\end{equation*}
where each $*$ denotes a scalar multiple of the power vector appearing in the same column. Observe that the \(A\) rows displayed above contain the power vectors
\[
\boldsymbol{\alpha}^{i-2A},\,
\boldsymbol{\alpha}^{i-2(A+1)},\,
\dots,\,
\boldsymbol{\alpha}^{i-2(s-1)},\,
\boldsymbol{\alpha}^{i-2B},\,
\boldsymbol{\alpha}^{i-2(B+1)},\,
\dots,\,
\boldsymbol{\alpha}^{i-2(A-1)}.
\]
For each of these rows, the displayed power vector lies in a degree layer which does not occur in any preceding row of the whole matrix. Hence none of these rows belongs to the span of the preceding rows. Therefore these rows contribute exactly $B+(A-B)=A$
new linearly independent rows.
        \item For $2t-s\le i\le 2t-2$, Lemma~\ref{Lem:2.2}(ii) shows that the rightmost $A\times A$ submatrix $
    \bigl(c_{\ell,j}\bigr)_{1 \le \ell \le A,\; B - A + 1 \le j \le B}$ is invertible. Hence, after elementary row operations,  \eqref{equ:case-i-2t-2s+1--2t-2} is row-equivalent to a matrix of the form
  \begin{equation*}
  \begin{pmatrix}
    * & \cdots & * & \boldsymbol{\alpha}^{i-2B} &
    \boldsymbol{0} & \cdots &
    \boldsymbol{0} \\
    * & \cdots & * & \boldsymbol{0} &
    \boldsymbol{\alpha}^{i-2(B+1)} & \cdots &
    \boldsymbol{0} \\
    \vdots & \ddots & \vdots & \vdots & \vdots & \ddots & \vdots \\
    * &\cdots & * & 
    \boldsymbol{0} &
    \boldsymbol{0} & \cdots &
    \boldsymbol{\alpha}^{i-2(s-1)}
  \end{pmatrix},
\end{equation*}
where each $*$ denotes a scalar multiple of the power vector appearing in the same column. Since these power vectors lie in degree layers which do not occur in any preceding row, this block contributes exactly $A$ new independent rows.
    \end{enumerate}
    Therefore, for $p\ge t$, we have $\dim(\widehat{\mathcal{C}})=(2t-2s+1)s.$
\end{enumerate}       
\end{proof}

Before we discuss the importance of the above result, we recall the property of the Schur square of a random linear code. 
\begin{proposition}[{\cite{randriambololona2015linear}}]Let $C$ be selected uniformly at random from the set of $[n(k),k]$-linear codes over $\mathbb{F}_{q}$. 
%Define $t:\mathbb{N}\rightarrow\mathbb{N}$, $t(k):=k(k+1)/2-n(k)$ and $s:\mathbb{N}\rightarrow\mathbb{N}$, $s(k):=n(k)-k(k+1)/2$. 
%If $n(k)\geq k(k+1)/2$, let $s(k):=n(k)-k(k+1)/2$, then there exist constants $\hat{\delta}\in \mathbb{R}_{>0}$ such that for all large enough $k$, the probability of $$\\dim_{\mathbb{F}_{q}}(\hat{C})=\frac{k(k+1)}{2}$$ is at least $1-2^{-\hat{\delta} s(k)}$.
    \begin{enumerate}
        \item  If $n(k)\leq k(k+1)/2$, let $t(k):=k(k+1)/2-n(k)$, then there exist constants $\epsilon,\delta\in \mathbb{R}_{>0}$ such that for all large enough $k$, the probability of $$\dim_{\mathbb{F}_{q}}(\widehat{C})=n(k)$$ is at least $1-2^{-\epsilon k}-2^{-\delta t(k)}$.
        \item If $n(k)\geq k(k+1)/2$, let $s(k):=n(k)-k(k+1)/2$, then there exists a constant $\hat{\delta}\in \mathbb{R}_{>0}$ such that for all large enough $k$, the probability of $$\dim_{\mathbb{F}_{q}}(\widehat{C})=\frac{k(k+1)}{2}$$ is at least $1-2^{-\hat{\delta} s(k)}$.
    \end{enumerate}
\end{proposition}

In other words, with high probability, for a random $[n,k]$ linear code  $\mathscr{R}$ over $\mathbb{F}_{q}$, we have $$\dim_{\mathbb{F}_{q}}(\widehat{\mathscr{R}})=\min\left\{n,\binom{k+1}{2}\right\}.$$
By Theorem~\ref{thm:dim-Schur-square-HRS}, HRS codes $HRS_{t}(\boldsymbol{\alpha},s)$ with parameters $p\geq t=2s\geq 4$ and $r\geq t+1$ have square dimension $\frac{t(t+1)}{2}$ which achieves that of random linear codes. So they are expected to resist the Schur-square attack.

%an important application in code-based cryptography is the construction of distinguishers for square-code attacks. By considering the efficiency of code-based cryptosystems, linear codes with strong algebraic structure would be preferable to reduce the key size. However, the strong algebraic structure often makes the dimension of the Schur square significantly smaller than the corresponding value of a random code, which leads to a  distinguisher of the public code from random codes. In view of this, Theorem~\ref{thm:dim-Schur-square-HRS} indicates that HRS codes with certain specific parameters are expected to resist the Schur-square attack.
To finish the paper, we give an example to illustrate our main result and propose some open problems based on the data. 
\begin{example}\label{example}
Let $p, m, r, s, t$ be positive integers with $p$ a prime. Let $\mathbb{F}_{p^m}$ be a finite field of size $p^m$. Randomly choose a $r$-tuple $\boldsymbol{\alpha}=(\alpha_{1},\cdots,\alpha_{r})\in(\mathbb{F}_{p^m})^r$ with pairwise distinct $\alpha_{1},\cdots,\alpha_{r}$. Let $\mathcal{C}=HRS_{t}(\boldsymbol{\alpha},s)$ be the corresponding HRS code. 
Table~\ref{tab:my-table} lists the dimensions of Schur squares $\hat{\mathcal{C}}$ under condition $r\geq 2t-2s+1$ computed by using MAGMA.
%\begin{equation*}
%    \begin{pmatrix}
%    f(\alpha_{1})&f(\alpha_{2})&\cdots&f(\alpha_{r})\\
%    \partial^{(1)} f(\alpha_{1})&\partial^{(1)}f(\alpha_{2})&\cdots&\partial^{(1)}f(\alpha_{r})\\
%    \cdots&\cdots&\ddots&\cdots\\
%    \partial^{(s-1)} f(\alpha_{1})&\partial^{(s-1)}f(\alpha_{2})&\cdots&\partial^{(s-1)}f(\alpha_{r})\\
%\end{pmatrix},
%\end{equation*}
% where $f\in\mathbb{F}_{p^m}[x]_{\leq t-1}$. 
 %Of course, if we vectorize an $s\times r$  matrix into an $rs$-dimensional vector by scanning from left to right and top to bottom, we obtain 
% $$
% \left((f\cdot g)(\boldsymbol{\alpha}),(\partial^{(1)}f\cdot\partial^{(1)}g)(\boldsymbol{\alpha}),\cdots,(\partial^{(s-1)}f\cdot\partial^{(s-1)}g)(\boldsymbol{\alpha})\right),
% $$
% where $(\partial^{(j)}f\cdot\partial^{(j)}g)(\boldsymbol{\alpha})=(\partial^{(j)}f(\alpha_{1})\cdot\partial^{(j)}g(\alpha_{1}),\cdots,\partial^{(j)}f(\alpha_{r})\cdot\partial^{(j)}g(\alpha_{r}))$ for all $0\leq j\leq s-1$.
 
The data from groups $3,4,7,8,15,16,19,20$  verify the correctness of the main result: Theorem~\ref{thm:dim-Schur-square-HRS}~(ii). That is, for $r\geq 2t-2s+1$ and $p\geq t\geq 2s$, the dimension of the Schur square $\widehat{\mathcal{C}}$ equals $(2t-2s+1)s$.
\end{example}

There is another interesting phenomenon behind the data given in Example~\ref{example}: under the condition $p\geq t$, most values of $\dim(\widehat{\mathcal{C}})$ equal $(2t-2s+1)s$ or $t(t+1)/2$ (mainly depend on $t\geq 2s$ or not). Based on the data in Table~\ref{tab:my-table}, we propose the following conjecture.

\begin{conjecture}
    For $r\geq 2t-2s+1$ and $t\leq \min\{p,2s,r\}$, 
     the dimension of the Schur square $\widehat{\mathcal{C}}$ equals $t(t+1)/2$.
\end{conjecture}

\begin{table}[]
 	\centering\begin{footnotesize}
    \begin{threeparttable}
 	\begin{tabular}{|c|c|c|c|c|c|c|c|}
    \hline
&($p,m,r,s,t$) &  $(2t-3s+2)s$ & $\dim(\hat{\mathcal{C}})$ & $(2t-2s+1)s$ & $t(t+1)/2$ &Is $p\geq t$? & Is $t\geq 2s$?\\
      \hline
     1& $(31,3,9,10,8)$ & $--$\tnote{1} & $36$ & $--$& $36$ &Yes& No \\ \hline
     2& $(31,3,20,10,16)$ & $40$ & $136$ & $130$ &$136$&Yes &No \\ \hline
     3&$(31,3,24,10,20)$ & $120$ & $210$ & $210$&$210$ &Yes &Yes($t=2s$) \\ \hline
     4&$(31,3,33,10,24)$ & $200$ & $290$ & $290$&$300$ &Yes&Yes \\ \hline
     5& $(31,3,15,14,12)$ & $--$ & $78$ & $--$& $78$ &Yes& No \\ \hline
     6& $(31,3,26,14,24)$ & $112$ & $300$ & $294$ &$300$&Yes &No \\ \hline
     7&$(31,3,30,14,28)$ & $224$ & $406$ & $406$&$406$ &Yes &Yes($t=2s$) \\ \hline
     8&$(31,3,35,14,30)$ & $280$ & $462$ & $462$&$465$ &Yes&Yes \\ \hline
     9& $(31,3,45,16,37)$ & $448$ & $580$ & $688$& $703$ &No&Yes \\ \hline
     10& $(31,3,67,16,37)$ & $448$ & $687$ & $688$& $703$ &No&Yes \\ \hline
     11& $(31,3,68,16,37)$ & $448$ & $688$ & $688$& $703$ &No&Yes \\ \hline
     12& $(31,3,79,16,37)$ & $448$ & $688$ & $688$& $703$ &No&Yes \\ \hline
      13& $(43,2,11,12,10)$ & $--$ & $55$ & $--$& $55$ &Yes& No \\ \hline
     14& $(43,2,18,12,16)$ & $--$ & $136$ & $108$ &$136$&Yes &No \\ \hline
     15&$(43,2,30,12,24)$ & $168$ & $300$ & $300$&$300$ &Yes &Yes($t=2s$) \\ \hline
     16&$(43,2,35,12,28)$ & $264$ & $396$ & $396$&$406$ &Yes&Yes \\ \hline

      17& $(43,2,17,18,16)$ & $--$ & $136$ & $--$& $136$ &Yes& No \\ \hline
     18& $(43,2,26,18,24)$ & $--$ & $300$ & $234$ &$300$&Yes &No \\ \hline
     19&$(43,2,38,18,36)$ & $360$ & $666$ & $666$&$666$ &Yes &Yes($t=2s$) \\ \hline
     20&$(43,2,46,18,40)$ & $540$ & $810$ & $810$&$820$ &Yes&Yes \\ \hline

     21& $(43,2,72,20,50)$ & $840$ & $1115$ & $1220$& $1275$ &No& Yes \\ \hline
     22& $(43,2,92,20,50)$ & $840$ & $1219$ & $1220$ &$1275$&No &Yes \\ \hline
     23&$(43,2,93,20,50)$ & $840$ & $1220$ & $1220$&$1275$ &No &Yes \\ \hline
     24&$(43,2,100,20,50)$ & $840$ & $1220$ & $1220$&$1275$ &No&Yes \\ \hline
 	\end{tabular}
    \begin{tablenotes}
\footnotesize
\item[1] The symbol ``$--$" denotes negative values where we are intersted in comparing ``$\dim(\widehat{\mathcal{C}})$" with ``$(2t-2s+1)s$" or ``$t(t+1)/2$".

\end{tablenotes}
\end{threeparttable}
 	\caption{The dimensions of Schur square codes $\widehat{\mathcal{C}}$ under condition $r\geq 2t-2s+1$}
 	\label{tab:my-table}
\end{footnotesize} \end{table}

    %   \begin{table}[]
 	%\centering\begin{footnotesize}
 	%%\begin{tabular}{|c|c|c|c|c|c|c|c|}
    %\hline
%&($p,m,r,s,t$) &  $(2t-3s+2)s$ & $\dim(\hat{\mathcal{C}})$ & $(2t-2s+1)s$ & $t(t+1)/2$ &Is $p>t\geq 2s$? & Is $r\geq 2t-2s+1$?\\
 %     \hline
  %   1& $(3,2,6,2,4)$ & $8$ & $9$ & $10$& $10$ &No& Yes \\ \hline
   %  2& $(7,3,37,7,25)$ & $217$ & $161$ & $259$ &$325$&No &Yes \\ \hline
    % 3&$(7,3,212,4,28)$ & $184$ & $196$ & $196$&$406$ &No &Yes\\ \hline
     %4&$(11,1,10,1,4)$ & $7$ & $7$ & $7$&$10$ &Yes &Yes \\ \hline
    % 5&$(11,2,103,10,40)$ & $520$ & $414$ & $610$&$820$ &No&Yes \\ \hline
   %  6&$(11,2,117,6,40)$ & $384$ & $414$ & $414$&$820$&No &Yes \\ \hline
   %  7&$(11,3,160,7,23)$ & $189$ & $210$ & $231$&$276$&No &Yes \\ \hline
   %  8&$(31,2,20,10,20)$ & $120$ & $200$ & $210$&$210$&Yes ($t=2s$)&No \\ \hline
   %  9&$(31,2,21,10,20)$ & $120$ & $210$ & $210$&$210$&Yes ($t=2s$)& Yes\\ \hline
   %  10&$(31,3,87,13,28)$ & $247$ & $403$ & $403$ &$406$&Yes &Yes \\ \hline
   %  11&$(31,3,87,15,34)$ & $375$ & $585$ & $585$ &$595$&No&Yes \\ \hline
   %  12&$(31,3,130,15,70)$ & $1455$ & $1663$ & $1665$& $2485$ &No&Yes\\ \hline
   %  13&$(43,2,88,15,30)$ & $255$ & $465$ & $465$ &$465$&Yes ($t=2s$)&Yes \\ \hline
   %  14&$(43,2,95,20,42)$ & $520$ & $900$ & $900$&$903$&Yes&Yes \\ \hline
   %  15&$(43,3, 100,25,55)$ & $925$ & $1474$ & $1525$ &$1540$&No&Yes\\ \hline
   %  16&$(43,3,44,20,42)$ & $520$ & $880$ & $900$ &$903$&Yes&No\\ \hline
   %  17&$(43,3,46,20,42)$ & $520$ & $900$ & $900$ & $903$&Yes&Yes \\ \hline
 	%\end{tabular}
 %	\caption{The dimensions of Schur square codes $\widehat{\mathcal{C}}$}
 %	\label{tab:my-table}
%\end{footnotesize} \end{table}

\section{Conclusion}
In this paper, we have mainly studied the dimensions of Schur square codes  of HRS codes. A lower bound and an upper bound were obtained. Moreover, we showed that when $p\geq t\geq 2s$ and $t\leq \frac{r+2s-1}{2}$, the square dimension of the HRS code $HRS_{t}(\boldsymbol{\alpha},s)$ reaches the upper bound $(2t-2s+1)s$. It is worth pointing out that for $p\geq t= 2s$ and $t\leq \frac{r+2s-1}{2}$ the results achieve the performance of random codes.  Unfortunately, our results are restricted to the choices of parameters. Therefore, extending the study of the dimensions of Schur squares of HRS codes for general parameters constitutes an interesting and important future work.

\appendices

\section{Proofs of Lemmas~\ref{Lem:2.1}, \ref{Lem:2.3} and~\ref{Lem:2.2}}

Before proving Lemmas~\ref{Lem:2.1}, \ref{Lem:2.3} and~\ref{Lem:2.2}, we first establish the following auxiliary lemma.

\begin{lemma}\label{Lem:A.1}
  For integers $i,k \ge 1$ and $j \ge 0$, we have
  \[
    \binom{i}{k}\binom{j}{k} - \binom{i-1}{k}\binom{j+1}{k}
    = \frac{j - i + 1}{k}\,\binom{i-1}{k-1}\binom{j}{k-1}.
  \]
\end{lemma}

\begin{proof}
\[
\begin{aligned}
  \binom{i}{k}\binom{j}{k} - \binom{i-1}{k}\binom{j+1}{k}
  &= \binom{i-1}{k-1}\binom{j}{k-1}
     \left(\frac{i(j-k+1)}{k^{2}} - \frac{(i-k)(j+1)}{k^{2}}\right) \\
  &= \binom{i-1}{k-1}\binom{j}{k-1}
     \cdot \frac{k(j-i+1)}{k^{2}} \\
  &= \frac{j - i + 1}{k}\,\binom{i-1}{k-1}\binom{j}{k-1}.
\end{aligned}
\]
This completes the proof.
\end{proof}

\begin{proof}[Proof of Lemma~\ref{Lem:2.1}]
We first treat the special case $a=0$, and then deduce the general case.

\medskip
\noindent\textbf{Step 1: the case $a=0$.}
We claim that
\begin{equation}\label{eq:M0b}
  \det(M(0,b,i,j))
  = \prod_{w=0}^{b-1}
      \frac{\displaystyle\prod_{u=1}^{b-w}\bigl(j - i + 1 - w - 2u\bigr)}
           {(b-w)!}.
\end{equation}

By definition, we know
%%\[
%  M(0,b,i,j)
%  =
 % \begin{pmatrix}
 %   1 & \binom{i}{1}\binom{j}{1} & \cdots & \binom{i}{b}\binom{j}{b} \\
%    1 & \binom{i+1}{1}\binom{j-1}{1} & \cdots & \binom{i+1}{b}\binom{j-1}{b} \\
%    \vdots & \vdots & \ddots & \vdots \\
%    1 & \binom{i+b}{1}\binom{j-b}{1} & \cdots & \binom{i+b}{b}\binom{j-b}{b}
%  \end{pmatrix},
%\]
%so
\[
  \det(M(0,b,i,j))
  =
  \left|
  \begin{array}{cccc}
    1 & \binom{i}{1}\binom{j}{1} & \cdots & \binom{i}{b}\binom{j}{b} \\
    1 & \binom{i+1}{1}\binom{j-1}{1} & \cdots & \binom{i+1}{b}\binom{j-1}{b} \\
    \vdots & \vdots & \ddots & \vdots \\
    1 & \binom{i+b}{1}\binom{j-b}{1} & \cdots & \binom{i+b}{b}\binom{j-b}{b}
  \end{array}
  \right|.
\]

For each $\ell = b,b-1,\dots,1$, perform the following row operation
\[
  R_{\ell+1} \leftarrow R_{\ell+1} - R_{\ell}.
\]
This does not change the determinant. In the first column this produces
\[
  (1,0,\dots,0)^{\mathsf T}.
\]
For the $(\ell+1)$-th row and the $(k+1)$-th column ($1 \le k \le b$), the new entry is
\[
  \binom{i+\ell}{k}\binom{j-\ell}{k}
  - \binom{i+\ell-1}{k}\binom{j-\ell+1}{k}.
\]
By applying Lemma~\ref{Lem:A.1} with $i \mapsto i+\ell$ and $j \mapsto j-\ell$, we obtain
\[
  \binom{i+\ell}{k}\binom{j-\ell}{k}
  - \binom{i+\ell-1}{k}\binom{j-\ell+1}{k}
  = \frac{j-i+1-2\ell}{k}\,\binom{i+\ell-1}{k-1}\binom{j-\ell}{k-1}.
\]
Thus, the determinant of the obtained matrix has the form
\begin{equation*}
    \begin{aligned}
         \det(M(0,b,i,j))&=
  \left|
  \begin{array}{cccc}
    1 & \binom{i}{1}\binom{j}{1} & \cdots & \binom{i}{b}\binom{j}{b} \\[2pt]
    0 & (j-i-1)\binom{i}{0}\binom{j-1}{0} & \cdots & \dfrac{j-i-1}{b}\binom{i}{b-1}\binom{j-1}{b-1} \\
    \vdots & \vdots & \ddots & \vdots \\
    0 & (j-i-2b+1)\binom{i+b-1}{0}\binom{j-b}{0} & \cdots & \dfrac{j-i-2b+1}{b}\binom{i+b-1}{b-1}\binom{j-b}{b-1}
  \end{array}
  \right|\\
  &=\frac{\prod\limits_{u=1}^{b}(j-i+1-2u)}{b!}\cdot \left|
  \begin{array}{cccc}
    \binom{i}{0}\binom{j-1}{0}&\binom{i}{1}\binom{j-1}{1}&\cdots& \binom{i}{b-1}\binom{j-1}{b-1}  \\
    \binom{i+1}{0}\binom{j-2}{0}&\binom{i+1}{1}\binom{j-2}{1}&\cdots& \binom{i+1}{b-1}\binom{j-2}{b-1}  \\
    \vdots&\vdots&\ddots&\vdots\\
    \binom{i+b-1}{0}\binom{j-b}{0}&\binom{i+b-1}{1}\binom{j-b}{1}&\cdots& \binom{i+b-1}{b-1}\binom{j-b}{b-1}  \\
  \end{array}\right|\\
  &=\frac{\prod\limits_{u=1}^{b}(j-i+1-2u)}{b!}\cdot\det(M(0,b-1,i,j-1)).
    \end{aligned}
\end{equation*}

%In each of the last $b$ rows, every entry in columns $2$ to $b+1$ contains the factor
%\((j-i+1-2\ell)\). Hence we can factor out
%\[
%  \prod_{\ell=1}^{b}(j-i+1-2\ell)
%  = \prod_{u=1}^{b}(j-i+1-2u)
%\]
%from the rows. Moreover, in each column $k=1,\dots,b$ (i.e., columns $2$ to $b+1$ of the full matrix) there is a factor $1/k$ coming from the denominator, so we can factor out
%\[
%  \prod_{k=1}^{b}\frac{1}{k} = \frac{1}{b!}
%\]
%from the columns. The remaining $b \times b$ block is exactly $M(0,b-1,i,j-1)$. Therefore we obtain the recurrence
%\[
 % \det(M(0,b,i,j))
  %= \frac{\displaystyle\prod_{u=1}^{b}(j-i+1-2u)}{b!}\,
%    \det(M(0,b-1,i,j-1)).
%\]

More generally, for any $0 \le w \le b-1$, the same argument on $M(0,b-w,i,j-w)$ yields
\[
  \frac{\det(M(0,b-w,i,j-w))}{\det(M(0,b-w-1,i,j-w-1))}
  = \frac{\displaystyle\prod_{u=1}^{b-w}(j-i+1-w-2u)}{(b-w)!}.
\]
Thus,
\begin{equation*}
    \begin{aligned}
    \det(M(0,b,i,j))
  &= \left(
      \prod_{w=0}^{b-1}
      \frac{\det(M(0,b-w,i,j-w))}{\det(M(0,b-w-1,i,j-w-1))}
    \right)
    \det(M(0,0,i,j-b))\\
&  = \prod_{w=0}^{b-1}
      \frac{\displaystyle\prod_{u=1}^{b-w}(j-i+1-w-2u)}{(b-w)!},    
    \end{aligned}
\end{equation*}
which proves \eqref{eq:M0b}.

\medskip
\noindent\textbf{Step 2: the general case $a>0$.}
Now consider $M(a,b,i,j)$ with $0 \le a \le b$. By definition, we have
\[
  M(a,b,i,j)
  =
  \begin{pmatrix}
    \binom{i}{a}\binom{j}{a} & \cdots & \binom{i}{b}\binom{j}{b} \\
    \vdots & \ddots & \vdots \\
    \binom{i+b-a}{a}\binom{j-b+a}{a} & \cdots & \binom{i+b-a}{b}\binom{j-b+a}{b}
  \end{pmatrix}.
\]

For each $\ell=0,1,\dots,b-a$ and each $k=a,a+1,\dots,b$, by using the standard identities
\[
  \binom{i+\ell}{k}
  = \binom{i+\ell}{a}\,
    \frac{\binom{i+\ell-a}{k-a}}{\binom{k}{a}}
  \,\,\mbox{and}\,\,
  \binom{j-\ell}{k}
  = \binom{j-\ell}{a}\,
    \frac{\binom{j-\ell-a}{k-a}}{\binom{k}{a}},
\]
we have
\[
  \binom{i+\ell}{k}\binom{j-\ell}{k}
  = \binom{i+\ell}{a}\binom{j-\ell}{a}\,
    \frac{\binom{i+\ell-a}{k-a}\binom{j-\ell-a}{k-a}}{\binom{k}{a}^{2}}.
\]
%From the $\ell$-th row (for $\ell=0,\dots,b-a$) we can factor out $\binom{i+\ell}{a}\binom{j-\ell}{a}$, and from the column corresponding to $k=a+t$ ($t=0,1,\dots,b-a$) we can factor out $1/\binom{a+t}{a}^{2} = 1/\binom{a+t}{t}^{2}$. After extracting these factors, the remaining matrix is exactly
%\[
 % M(0,b-a,i-a,j-a).
%\]
Consequently, by extracting the factor $\binom{i+\ell}{a}\cdot\binom{j-\ell}{a}$  from the  $\ell$-th row and the factor  $\frac{1}{\binom{a+h}{a}^2}$ from the  $h$-th column of matrix   $M(a,b,i,j)$ for each $0\leq\ell,h\leq b-a$, we obtain  .
\[
  \det(M(a,b,i,j))
  = \left(
      \prod_{\ell=0}^{b-a}\binom{i+\ell}{a}\binom{j-\ell}{a}
    \right)\cdot
    \left(
      \prod_{\ell=0}^{b-a}\frac{1}{\binom{a+\ell}{\ell}^{2}}
    \right)
    \det(M(0,b-a,i-a,j-a).
\]
Finally, by applying the result of Step~1, we obtain
\begin{align*}
    \det(M(a,b,i,j))
 & = \prod_{\ell=0}^{b-a}
      \frac{\binom{i+\ell}{a}\binom{j-\ell}{a}}{\binom{a+\ell}{\ell}^{2}}
    \cdot \det(M(0,b-a,i-a,j-a))\\
    &= \left(\prod_{\ell=0}^{b-a}
      \frac{\binom{i+\ell}{a}\binom{j-\ell}{a}}{\binom{a+\ell}{\ell}^{2}}\right)
    \cdot
    \left(\prod_{w=0}^{b-a-1}
      \frac{\displaystyle\prod_{u=1}^{b-a-w}\bigl(j - i + 1 - w - 2u\bigr)}
           {(b-a-w)!}\right),
\end{align*}
which completes the proof.

%Now we apply the result of Step~1 (with $b$ replaced by $b-a$ and $j$ replaced by $j-a$) to $\det(M(0,b-a,i-a,j-a))$, and obtain
%\[
%  \det(M(0,b-a,i-a,j-a))
%  = \prod_{w=0}^{b-a-1}
%      \frac{\displaystyle\prod_{u=1}^{b-a-w}\bigl(j-i+1-w-2u\bigr)}
%           {(b-a-w)!}.
%\]
%Substituting this into the previous expression yields
%\[
 % \det(M(a,b,i,j))
 % = \left(\prod_{\ell=0}^{b-a}
  %    \frac{\binom{i+\ell}{a}\binom{j-\ell}{a}}{\binom{a+\ell}{\ell}^{2}}\right)
   % \cdot
   % \left(\prod_{w=0}^{b-a-1}
    %  \frac{\displaystyle\prod_{u=1}^{b-a-w}\bigl(j - i + 1 - w - 2u\bigr)}
     %      {(b-a-w)!}\right)
%\]

\end{proof}

\begin{proof}[Proof of Lemma~\ref{Lem:2.3}]
We split the proof into two parts. First we derive a recursive relation
for $\det(M_{\ell,k}(0,b,i,j))$ when $\ell \ge 2$, and then we evaluate
the base case $\ell=1$. Combining these two steps yields the desired
closed formulas.

\medskip
\noindent\textbf{Step 1: A recursion for $\ell \ge 2$.}
Assume $\ell\ge 2$.
By applying the same row operations as in the proof of Lemma~\ref{Lem:2.1}
and using Lemma~\ref{Lem:A.1} at each step, one obtains
\[
  \det\bigl(M_{\ell,k}(0,b,i,j)\bigr)
  =
  \frac{\displaystyle\prod_{u=1}^{b}(j-i-2u+1)}{
        b!\,\dfrac{k}{\ell-1}}\,
  \det\bigl(M_{\ell-1,k-1}(0,b-1,i,j-1)\bigr),
\]
for every integer $\ell$ with $2 \le \ell \le b+1$ and $k \ge b+1$.

Iterating this relation, for each $0 \le m \le \ell-2$ we obtain
\[
  \frac{\det\bigl(M_{\ell-m,k-m}(0,b-m,i,j-m)\bigr)}%
       {\det\bigl(M_{\ell-m-1,k-m-1}(0,b-m-1,i,j-m-1)\bigr)}
  =
  \frac{\displaystyle\prod_{u=1}^{b-m}
        (j-m-i-2u+1)}{(b-m)!\,\dfrac{k-m}{\ell-m-1}}.
\]
Hence, by telescoping the above products, we get
\begin{equation}\label{eq:recursion-ell}
\begin{aligned}
  \det\bigl(M_{\ell,k}(0,b,i,j)\bigr)
  &= \prod_{m=0}^{\ell-2}
       \frac{\det\bigl(M_{\ell-m,k-m}(0,b-m,i,j-m)\bigr)}%
            {\det\bigl(M_{\ell-m-1,k-m-1}(0,b-m-1,i,j-m-1)\bigr)} \\
  &\qquad\cdot\det\left(M_{1,k-\ell+1}(0,b-\ell+1,i,j-\ell+1)\right) \\
  &= \left(\prod_{m=0}^{\ell-2}
       \frac{\displaystyle\prod_{u=1}^{b-m}(j-m-i-2u+1)}%
            {(b-m)!\,\dfrac{k-m}{\ell-m-1}}\right)\cdot
     \det\bigl(M_{1,k-\ell+1}(0,b-\ell+1,i,j-\ell+1)\bigr).
\end{aligned}
\end{equation}
Thus, for $\ell\ge 2$, it remains to compute the determinant
$\det(M_{1,K}(0,B,I,J))$ in closed form, where 
\[
  B = b-\ell+1,\quad
  I = i,\quad
  J = j-\ell+1,\quad
  K = k-\ell+1.
\]

\medskip
\noindent\textbf{Step 2: The case $\ell=1$.}
We now evaluate $\det(M_{1,K}(0,B,I,J))$ for general integers
$B,I,J,K$ with $0\le B\le p-1$, $I\ge 0$, $J\ge B+1$ and $K\ge B+1$.

By definition, $M_{1,K}(0,B,I,J)$ is obtained from the matrix
$M(0,B,I,J)$ by replacing its first column with the column corresponding
to the index $K$:
\[
  M_{1,K}(0,B,I,J)
  =
  \begin{pmatrix}
    \binom{I}{K}\binom{J}{K} &
    \binom{I}{1}\binom{J}{1} & \cdots &
    \binom{I}{B}\binom{J}{B} \\
    \vdots & \vdots & & \vdots \\
    \binom{I+B}{K}\binom{J-B}{K} &
    \binom{I+B}{1}\binom{J-B}{1} & \cdots &
    \binom{I+B}{B}\binom{J-B}{B}
  \end{pmatrix}.
\]
By permuting the first column of matrix $M_{1,K}(0,B,I,J)$ to the last column, we obtain 
%Permuting columns so that the $K$-column becomes the last one introduces
%a factor $(-1)^{B}$ in the determinant. Thus
\[
  \det\bigl(M_{1,K}(0,B,I,J)\bigr)
  = (-1)^{B} \det(\widehat{M}),
\]
where $\widehat{M}$ is the matrix with columns ordered as
$1,2,\dots,B,K$.

Next, as in the proof of Lemma~\ref{Lem:2.1}, we factor out $(I+h)(J-h)$  from $h$-th row   and $g^2$ 
   from  $g$-th column, where $h=0,1,\cdots,B$  and $g=1,2,\cdots,B,K$. Thus, we obtain

%Next, as in the proof of Lemma~\ref{Lem:2.1}, we factor out from each
%row $r$ (for $0\le r\le B$) the product $(I+r)(J-r)$, and from each
%non–distinguished column we factor out suitable binomial coefficients.
%A direct computation shows that this yields
%\[
%  \det(\widehat{M})
%  =
%  \frac{\displaystyle\prod_{h=0}^{B}(I+h)(J-h)}{(B!)^{2}K^{2}}
%  \det(\widetilde{M}),
%\]
\[
  \det(\widehat{M})
  =
  \frac{\displaystyle\prod_{h=0}^{B}(I+h)(J-h)}{(B!)^{2}K^{2}}
  \det(M_{B+1,K-1}(0,B,I-1,J-1)).
\]
%where $\widetilde{M}$ is again a $(B+1)\times(B+1)$ matrix of the same
%type as in Lemma~\ref{Lem:2.1}, but with parameters $(I-1,J-1)$ and with
%the distinguished column index reduced by $1$.

In general, for $0\leq u\leq B-1$, we have
\begin{equation*}
    \frac{\det(M_{B+1-u,K-1-u}(0,B-u,I-1,J-1-u))}{\det(M_{B-u,K-2-u}(0,B-1-u,I-1,J-u-2))}=\prod\limits_{h=0}^{B-u-1}\frac{J-I-u-2h-1}{(B-u-1)!(K-u-1)}.
\end{equation*}
Thus, we have
\small{
\begin{equation}\label{eq:l1-general}
   \begin{aligned}
     &\det\bigl(M_{1,K}(0,B,I,J)\bigr)
  = (-1)^{B} \det(\widehat{M})=(-1)^B\frac{\displaystyle\prod_{h=0}^{B}(I+h)(J-h)}{(B!)^{2}K^{2}}
  \det(M_{B+1,K-1}(0,B,I-1,J-1))\\
  =&(-1)^B\frac{\displaystyle\prod_{h=0}^{B}(I+h)(J-h)}{(B!)^{2}K^{2}}\cdot \prod\limits_{u=0}^{B-1}\frac{\det(M_{B+1-u,K-1-u}(0,B-u,I-1,J-1-u))}{\det(M_{B-u,K-2-u}(0,B-1-u,I-1,J-u-2))}\cdot\binom{I-1}{K-B-1}\binom{J-B-1}{K-B-1}\\
  =&(-1)^{B}\frac{\prod\limits_{h=0}^{B}(I+h)(J-h)}{(B!)^{2}K^{2}} \cdot
     \prod\limits_{u=0}^{B-1}
       \frac{\prod\limits_{h=0}^{B-1-u}
               \bigl(J-u-I-1-2h\bigr)}{(B-u-1)!\,(K-u-1)}\cdot
     \binom{I-1}{K-B-1}\binom{J-B-1}{K-B-1}.
\end{aligned} 
\end{equation}
}
%Repeating the same reduction with respect to the parameter $B$ and using
%Lemma~\ref{Lem:A.1} at each step, one obtains a recursion analogous to
%Step~1. Iterating it down to the $2\times 2$ base case and evaluating
%the resulting $2\times 2$ determinant explicitly, we finally arrive at
%the closed formula
%\begin{equation}\label{eq:l1-general}
%\begin{aligned}
 % \det\bigl(M_{1,K}(0,B,I,J)\bigr)
 % &= (-1)^{B}\,
  %   \left(\frac{\displaystyle\prod_{h=0}^{B}(I+h)(J-h)}{(B!)^{2}K^{2}}\right) \cdot\left(
   %  \prod_{u=0}^{B-1}
    %   \frac{\displaystyle\prod_{h=0}^{B-1-u}
 %              \bigl(J-u-I-1-2h\bigr)}{(B-u-1)!\,(K-u-1)} \right)\\
 % &\quad\cdot
  %   \binom{I-1}{K-B-1}\binom{J-B-1}{K-B-1}.
%\end{aligned}
%\end{equation}
This is exactly the desired expression for Lemma~\ref{Lem:2.3} in the
case $\ell=1$, once we substitute $B=b$, $I=i$, $J=j$ and $K=k$.

\medskip
\noindent\textbf{Step 3: Combination.}
We now return to the general case $\ell\ge 2$. In the recursion
\eqref{eq:recursion-ell}, we insert the formula
\eqref{eq:l1-general} for
\[
  \det\bigl(M_{1,k-\ell+1}(0,b-\ell+1,i,j-\ell+1)\bigr)
\]
with
\[
  B=b-\ell+1,\quad I=i,\quad J=j-\ell+1,\quad K=k-\ell+1.
\]
A straightforward simplification of the resulting products yields
\begin{equation*}
\begin{aligned}
  \det\bigl(M_{\ell,k}(0,b,i,j)\bigr)
  &= (-1)^{\,b-\ell+1}
     \left(\prod_{m=0}^{\ell-2}
       \frac{
         \displaystyle\prod_{u=1}^{b-m}\bigl(j-m-i-2u+1\bigr)
       }{
         (b-m)!\,\dfrac{k-m}{\ell-m-1}
       }\right)
     \cdot\left(
     \frac{
       \displaystyle\prod_{h=0}^{b-\ell+1}(i+h)(j-\ell+1-h)
     }{
       \bigl((b-\ell+1)!\bigr)^{2}\,(k-\ell+1)^{2}
     } \right)\\
  &\quad\cdot \left(
     \prod_{u=0}^{b-\ell}
       \frac{
         \displaystyle\prod_{h=0}^{b-\ell-u}\bigl(j-u-i-\ell-2h\bigr)
       }{
         (b-\ell-u)!\,(k-\ell-u)
       }\right)
     \cdot
     \binom{i-1}{k-b-1}\binom{j-b-1}{k-b-1},
\end{aligned}
\end{equation*}
which is precisely the formula stated in Lemma~\ref{Lem:2.3} for
$\ell\ge 2$.

Together with the case $\ell=1$ given by \eqref{eq:l1-general}, this
completes the proof of Lemma~\ref{Lem:2.3}.
\end{proof}

\begin{proof}[Proof of Lemma~\ref{Lem:2.2}]
Let $A= t - \Bigl\lceil \frac{i}{2} \Bigr\rceil$ and 
  $B= s - t + \Bigl\lceil \frac{i}{2} \Bigr\rceil$,
then $A+B=s$.  From Lemma~\ref{Lem:2.1}, we have
\[
  \det\bigl(M(0,A-1,i-t+1,t-1)\bigr)
  =
  \prod_{w=0}^{A-2}
    \frac{\displaystyle\prod_{u=1}^{A-1-w}\bigl(2t-i-1-w-2u\bigr)}
         {(A-1-w)!}.
\]
Since
\[
  1 \le 2t-i-1-w-2u \le 2s-4 < p
\]
for all $0 \le w \le A-2$ and $1 \le u \le A-1-w$, none of these factors vanishes in
$\mathbb{F}_{p}$, and hence
\[
  \det\bigl(M(0,A-1,i-t+1,t-1)\bigr) \neq 0.
\]
Thus $M(0,A-1,i-t+1,t-1)$ is invertible.

\medskip
\noindent\textbf{Case (i): $2t-2s+1 \le i \le 2t-s-1$.}
In this range we have $A \ge B$.
Set
\[
  D := M(0,A-1,i-t+1,t-1),
\]
and for $1 \le r,c \le A$ let $D_{r,c}$ be the minor obtained by deleting the
$r$-th row and $c$-th column from $D$. Then the inverse of $D$ is given by
\[
  D^{-1}
  = \frac{1}{\det(D)}
    \begin{pmatrix}
      D_{1,1} & \cdots & (-1)^{A+1}D_{A,1} \\
      \vdots  & \ddots & \vdots           \\
      (-1)^{A+1}D_{1,A} & \cdots & D_{A,A}
    \end{pmatrix},
\]
i.e. the transpose of the cofactor matrix divided by $\det(D)$.

By the definition in Lemma~\ref{Lem:2.2}, the matrix
$C=(c_{\ell,j})_{1\le \ell\le A,\;1\le j\le B}$ satisfies
\[
  D \, C = M(A,s-1,i-t+1,t-1,A).
\]
Hence
\[
  C = D^{-1} M(A,s-1,i-t+1,t-1,A).
\]
Writing this out explicitly, we obtain
\begin{equation*}
\begin{aligned}
  &\begin{pmatrix}
    c_{1,1} & \cdots & c_{1,B} \\
    \vdots  &        & \vdots  \\
    c_{A,1} & \cdots & c_{A,B}
  \end{pmatrix}\\
  =& \frac{1}{\det(D)}
     \begin{pmatrix}
       D_{1,1} & \cdots & (-1)^{A+1}D_{A,1} \\
       \vdots  & \ddots & \vdots           \\
       (-1)^{A+1}D_{1,A} & \cdots & D_{A,A}
     \end{pmatrix}\cdot
     \begin{pmatrix}
       \binom{i-t+1}{A}\binom{t-1}{A} & \cdots &
       \binom{i-t+1}{s-1}\binom{t-1}{s-1} \\
       \vdots & \ddots & \vdots \\
       \binom{i-t+A}{A}\binom{t-A}{A} & \cdots &
       \binom{i-t+A}{s-1}\binom{t-A}{s-1}
     \end{pmatrix}\\
     =&\frac{1}{\det(D)}\begin{pmatrix}
         \sum\limits_{h=1}^{A}(-1)^{1+h}D_{h,1}\cdot\binom{i-t+h}{A}\binom{t-h}{A}&\cdots&\sum\limits_{h=1}^{A}(-1)^{1+h}D_{h,1}\cdot\binom{i-t+h}{s-1}\binom{t-h}{s-1}\\
         \sum\limits_{h=1}^{A}(-1)^{2+h}D_{h,2}\cdot\binom{i-t+h}{A}\binom{t-h}{A}&\cdots&\sum\limits_{h=1}^{A}(-1)^{2+h}D_{h,2}\cdot\binom{i-t+h}{s-1}\binom{t-h}{s-1}\\
         \vdots&\ddots&\vdots\\
         \sum\limits_{h=1}^{A}(-1)^{A+h}D_{h,A}\cdot\binom{i-t+h}{A}\binom{t-h}{A}&\cdots&\sum\limits_{h=1}^{A}(-1)^{A+h}D_{h,A}\cdot\binom{i-t+h}{s-1}\binom{t-h}{s-1}\\
     \end{pmatrix}\\
     =&\frac{1}{\det(D)}\cdot\begin{pmatrix}
        \det(M_{1,A}(0,A-1,i-t+1,t-1))&\cdots&\det(M_{1,s-1}(0,A-1,i-t+1,t-1))\\
         \det(M_{2,A}(0,A-1,i-t+1,t-1))&\cdots&\det(M_{2,s-1}(0,A-1,i-t+1,t-1))\\
         \vdots&\ddots&\vdots\\
         \det(M_{A,A}(0,A-1,i-t+1,t-1))&\cdots&\det(M_{A,s-1}(0,A-1,i-t+1,t-1))
     \end{pmatrix}.
\end{aligned}
\end{equation*}

% Expanding along the column replaced in $M_{\ell,k}(0,A-1,i-t+1,t-1)$ and using
%  Cramer's rule, we see that for $1\le \ell \le A$ and $A \le k \le s-1$ the $(\ell,k)$-entry in the above product is precisely
%\[
%  \frac{1}{\det(D)} \det\bigl(M_{\ell,k}(0,A-1,i-t+1,t-1)\bigr).
%\]
%For brevity, we now write
%\[
 % M_{\ell,k}(0,A-1,i-t+1,t-1)
%\]
%to denote this determinant. Thus
%\[
%\small{ \begin{pmatrix}
%    c_{1,1} & \cdots & c_{1,B} \\
%    \vdots  &        & \vdots  \\
%    c_{A,1} & \cdots & c_{A,B}
%  \end{pmatrix}
%  = \frac{1}{\det(D)}
%    \begin{pmatrix}
%      M_{1,A}(0,A-1,i-t+1,t-1) & \cdots & M_{1,s-1}(0,A-1,i-t+1,t-1) \\
%      \vdots & \ddots & \vdots \\
%      M_{B,A}(0,A-1,i-t+1,t-1) & \cdots & M_{B,s-1}(0,A-1,i-t+1,t-1)
%    \end{pmatrix}}.
% \]
Therefore, to show that the $B\times B$ leading principal submatrix
$\bigl(c_{\ell,j}\bigr)_{1\le \ell\le B,\;1\le j\le B}$ is invertible, it suffices
to prove that
\[
  \begin{pmatrix}
   \det(M_{1,A}(0,A-1,i-t+1,t-1)) & \cdots & \det(M_{1,s-1}(0,A-1,i-t+1,t-1)) \\
    \vdots & \ddots & \vdots \\
    \det(M_{B,A}(0,A-1,i-t+1,t-1)) & \cdots & \det(M_{B,s-1}(0,A-1,i-t+1,t-1))
  \end{pmatrix}
\]
is invertible.

By Lemma~\ref{Lem:2.3}, for all $2 \le \ell \le B$ and $0 \le w \le B-1$ we have
\begin{equation*}
\begin{aligned}
  &\det(M_{\ell,A+w}(0,A-1,i-t+1,t-1)) \\
  &= (-1)^{A-\ell}
     \left(\prod_{m=0}^{\ell-2}
       \frac{\displaystyle\prod_{u=1}^{A-m-1}\bigl(2t-m-i-2u-1\bigr)}
            {(A-m-1)!\,\dfrac{A+w-m}{\ell-m-1}}\right)
     \cdot
     \left(\frac{\displaystyle\prod_{h=0}^{A-\ell}(i-t+1+h)(t-\ell-h)}
          {((A-\ell)!)^{2}(A+w-\ell+1)^{2}}\right)       \\
  &\quad\cdot\left(
     \prod_{u=0}^{A-\ell-1}
       \frac{\displaystyle\prod_{h=0}^{A-\ell-u-1}\bigl(2t-u-i-\ell-2h-2\bigr)}
            {(A-\ell-u-1)!(A+w-\ell-u)}\right)
     \cdot \binom{i-t}{w}\binom{t-A-1}{w}.
\end{aligned}
\end{equation*}
Rearranging the factors, this can be written as
\begin{equation*}
\begin{aligned}
  &\det(M_{\ell,A+w}(0,A-1,i-t+1,t-1))\\
 = & (-1)^{A-\ell}
     \prod_{m=0}^{\ell-2}
       \frac{
        (\ell-m-1) \displaystyle\prod_{u=1}^{A-m-1}\bigl(2t-m-i-2u-1\bigr)
       }{(A-m-1)!}\cdot
     \frac{\displaystyle\prod_{h=0}^{A-\ell}(i-t+1+h)(t-\ell-h)}
          {((A-\ell)!)^{2}}\\
  \cdot&
     \prod_{u=0}^{A-\ell-1}
       \frac{\displaystyle\prod_{h=0}^{A-\ell-u-1}\bigl(2t-u-i-\ell-2h-2\bigr)}
            {(A-\ell-u-1)!}\cdot \binom{i-t}{w}\binom{t-A-1}{w} \\
   \cdot&  \left(
      (A+w-\ell+1)^{2} \prod_{m=0}^{\ell-2}(A+w-m)
       \prod_{u=0}^{A-\ell-1}(A+w-\ell-u)
     \right)^{-1}.
\end{aligned}
\end{equation*}
Moreover,
\begin{equation*}
\begin{aligned}
  &\left(
     (A+w-\ell+1)^{2}\prod_{m=0}^{\ell-2}(A+w-m)
     \prod_{u=0}^{A-\ell-1}(A+w-\ell-u)
   \right)^{-1}               \\
  &= \left(
       (A+w-\ell+1)
       \prod_{u=0}^{A-1}(A+w-u)
     \right)^{-1}
   = \left(\prod_{u=0}^{A-1}(A+w-u)\right)^{-1}
     \cdot \frac{1}{A+w-\ell+1}.
\end{aligned}
\end{equation*}

Factoring out the corresponding common factors from each row $\ell$ and column index $w$, we obtain
\small{\begin{equation*}
\begin{aligned}
  &\det
   \begin{pmatrix}
     \det(M_{1,A}(0,A-1,i-t+1,t-1)) & \cdots & \det(M_{1,s-1}(0,A-1,i-t+1,t-1)) \\
     \vdots & \ddots & \vdots \\
     \det(M_{B,A}(0,A-1,i-t+1,t-1)) & \cdots & \det(M_{B,s-1}(0,A-1,i-t+1,t-1))
   \end{pmatrix}                       \\
  =& \prod_{\ell=1}^{B}
       \left(
         (-1)^{A-\ell}
         \frac{\displaystyle\prod_{h=0}^{A-\ell}(i-t+1+h)(t-\ell-h)}
              {((A-\ell)!)^{2}}
         \prod_{u=0}^{A-\ell-1}
           \frac{\displaystyle\prod_{h=0}^{A-\ell-u-1}(2t-u-i-\ell-2h-2)}
                {(A-\ell-u-1)!}
       \right)                          \\
  \cdot&
     \prod_{\ell=2}^{B}\;
       \prod_{m=0}^{\ell-2}
         \frac{
           (\ell-m-1)\displaystyle\prod_{u=1}^{A-m-1}\bigl(2t-m-i-2u-1\bigr)
         }{(A-m-1)!}\cdot
     \prod_{w=0}^{B-1}
       \frac{\displaystyle\binom{i-t}{w}\binom{t-A-1}{w}}
            {\displaystyle\prod_{u=0}^{A-1}(A+w-u)}\\
            &\cdot
     \det
     \begin{pmatrix}
       \frac{1}{A}         & \frac{1}{A+1}     & \cdots & \frac{1}{A+B-1} \\
       \frac{1}{A-1}       & \frac{1}{A}       & \cdots & \frac{1}{A+B-2} \\
       \vdots              & \vdots            & \ddots & \vdots          \\
       \frac{1}{A-B+1}     & \frac{1}{A-B+2}   & \cdots & \frac{1}{A}
     \end{pmatrix}.
\end{aligned}
\end{equation*}}

The last matrix is a Cauchy matrix with entries
\[
  \frac{1}{(A-\ell+1)+(j-1)},\qquad 1\le \ell,j\le B.
\]
Since $A\ge B$ and $A+B=s$, the denominator ranges over the integers
\[
  A-B+1,\dots,A+B-1 = s-1,
\]
all of which lie in $\{1,\dots,s-1\}\subseteq\{1,\dots,p-1\}$ (because $p\ge  2s$).
Thus all denominators are nonzero in $\mathbb{F}_{p}$, and by Cauchy’s
determinant formula the last determinant is nonzero. Therefore the whole product
above is nonzero, and the $B\times B$ matrix
$\bigl(c_{\ell,j}\bigr)_{1\le \ell\le B,\;1\le j\le B}$ is invertible.
This proves part~(i).

\medskip
\noindent\textbf{Case (ii): $2t-s \le i \le 2t-2$.}
In this range one has $A \le B$. We need to prove that the $A\times A$ submatrix $\bigl(c_{\ell,j}\bigr)_{1\le \ell\le A,\;B-A+1\le j\le B}$ is invertible. Since $D$ is invertible, it is enough to show that
$ M(B,s-1,i-t+1,t-1)$
is invertible. 
%Because the last $A$ columns of $C$ are obtained by multiplying
%$D^{-1}$ with the last $A$ columns of $M(A,s-1,i-t+1,t-1,A)$, which %form $M(B,s-1,i-t+1,t-1)$.

%We distinguish several subcases.

First, since $2t-s \le 2t-2$, we have $s \ge 2$.  If $s=2$, then necessarily $i=2t-2$ and
\[
 \det( M(1,1,t-1,t-1))
  =
  \binom{t-1}{1}\binom{t-1}{1}\neq 0.
\]

If $s\ge 3$ and $2t-3 \le i \le 2t-2$, then
 $ B = s - t + \left\lceil\frac{i}{2}\right\rceil = s-1.$
and therefore
\[
  \det\bigl(M(B,s-1,i-t+1,t-1)\bigr)
  = \det\bigl(M(s-1,s-1,\lfloor i/2 \rfloor,t-1)\bigr)
  = \binom{\lfloor i/2 \rfloor}{s-1}\binom{t-1}{s-1} \neq 0.
\]

Finally, assume $s\ge 4$ and $2t-s \le i \le 2t-4$. By Lemma~\ref{Lem:2.1}, we have
\begin{equation*}
\begin{aligned}
  &\det\bigl(M(B,s-1,i-t+1,t-1)\bigr)\\
  =&
  \prod_{\ell=0}^{s-1-B}
    \frac{\displaystyle\binom{i-t+1+\ell}{B}\binom{t-1-\ell}{B}}
         {\displaystyle\binom{B+\ell}{\ell}^{2}}                    \cdot
    \prod_{w=0}^{s-B-2}
      \frac{\displaystyle\prod_{u=1}^{s-B-1-w}(2t-i-1-w-2u)}
           {(s-1-B-w)!}.
\end{aligned}
\end{equation*}
Using the relations
\[
  A = t - \Bigl\lceil\frac{i}{2}\Bigr\rceil,
  \qquad
  B = s - t + \Bigl\lceil\frac{i}{2}\Bigr\rceil,
\]
and the bounds $2t-s \le i \le 2t-4$, $s\ge 4$, $t\ge 2s$, one checks that
the following inequalities hold:

\begin{itemize}
  \item For all $0\le \ell\le s-1-B$,
\begin{gather*}
  2\le \Bigl\lfloor\frac{i}{2}\Bigr\rfloor - s +2
     \le i-t+1+\ell-(B-1)
     \le i-t+1+\ell
     \le \Bigl\lfloor\frac{i}{2}\Bigr\rfloor < t \leq p,\\
  1 < t-s+1
     \le t-1-\ell-(B-1)
     < t-1-\ell
     < t \leq p,\\
  2 \le \left\lceil\frac{s}{2}\right\rceil
     \le s - t + \left\lceil\frac{i}{2}\right\rceil
     \le B+\ell
     \le s-1 < t \leq p.
\end{gather*}
  
  \item For all $0\le w\le s-B-2$ and $1\le u\le s-1-B-w$,
  \[
    1 \le 2\left\lceil\frac{i}{2}\right\rceil - i +1
      \le 2t-i-1-(2s-2B-2)
      \le 2t-i-1-w-2u
      \le 2t-i-3
      \le s-3.
  \]
\end{itemize}

Thus every binomial coefficient and every factor $2t-i-1-w-2u$ appearing in the
above product is an integer in the range $\{1,\dots,p-1\}$, and hence is
nonzero in $\mathbb{F}_{p}$. It follows that
\[
  \det\bigl(M(B,s-1,i-t+1,t-1)\bigr) \neq 0.
\]

This shows that in Case~(ii) the $A\times A$ submatrix
$\bigl(c_{\ell,j}\bigr)_{1\le \ell\le A,\;B-A+1\le j\le B}$ is invertible, which
completes the proof of Lemma~\ref{Lem:2.2}.
\end{proof}

\bibliographystyle{plain}
\bibliography{schurHRS}

\end{document}